\documentclass[12pt, letterpaper]{amsart}

\usepackage{amsthm,amsmath,amsfonts,amssymb,amscd,mathrsfs,amsaddr}
\usepackage{txfonts}
\usepackage{supertabular,soul}
\usepackage{tikz, graphicx,color,geometry}
\usepackage{multirow}

\usepackage{bbm} 
\usepackage{yfonts}
\usepackage{eucal}
\usepackage{overpic}

\newtheorem{theorem}{Theorem}
\newtheorem{definition}{Definition}
\newtheorem{thm}{Theorem}[section]

\theoremstyle{definition}

\newtheorem{example}[thm]{Example}

\theoremstyle{remark}

\let\oldmarginpar\marginpar
\renewcommand\marginpar[1]{\oldmarginpar[\raggedleft\footnotesize #1]%
{\raggedright\footnotesize #1}}

\begin{document}

\title[Specialization Models of Network Growth]{Specialization Models of Network Growth}

\keywords{Specialization, Network Growth, Modular and Hierarchical Structure, Intrinsic Stability}

\author{L. A. Bunimovich$^1$}
\address{School of Mathematics, Georgia Institute of Technology, 686 Cherry Street, Atlanta, GA 30332}
\email{bunimovich@math.gatech.edu}
\author{D. C. Smith$^2$}
\address{Department of Mathematics, Brigham Young University, TMCB 274, Provo, UT 84602}
\email{dallas.smith@mathematics.byu.edu}
\author{B. Z. Webb$^3$}
\address{Department of Mathematics, Brigham Young University, TMCB 274, Provo, UT 84602}
\email{bwebb@mathematics.byu.edu}

\begin{abstract}
One of the most important features observed in real networks is that, as a network's topology evolves so does the network's ability to perform various complex tasks. To explain this, it has also been observed that as a network grows certain subnetworks begin to specialize the function(s) they perform. Here, we introduce a class of models of network growth based on this notion of \emph{specialization} and show that as a network is specialized using this method its topology becomes increasingly sparse, modular, and hierarchical, each of which are important properties observed in real networks. This procedure is also highly flexible in that a network can be specialized over any subset of its elements. This flexibility allows those studying specific networks the ability to search for mechanisms that describe the growth of these particular networks. As an example, we find that by randomly selecting these elements a network's topology acquires some of the most well-known properties of real networks including the small-world property, disassortativity, power-law like degree distributions, and power-law like clustering coefficients. As far as the authors know, this is the first such class of models that creates an increasingly modular and hierarchical network topology with these properties.
\end{abstract}

\maketitle

\section{Introduction}

Networks studied in the biological, social, and technological sciences perform various tasks, which are determined by both the network's topology as well as the network's dynamics. In the biological setting gene regulatory and metabolic networks allow cells to organize into tissues and tissues into organs whose dynamics are essential to the network's function \cite{KS08,BO04}, e.g. a beating heart in a circulatory network. Neuronal networks are responsible for complicated processes related to cognition and memory, which are based on the network's structure of connections as well as the electrical dynamics of the network's neurons \cite{BS09}. Social networks such as Facebook, Twitter, the interactions of social insects \cite{CF14,HEOGF13}, and professional sports teams \cite{FAIPW12} function as a collection of overlapping communities or a single unified whole based on the underlying topology and hierarchies present within the network's social interactions. Technological networks such as the internet together with the World Wide Web allow access to information based on the topology of network links and the network's dynamic ability to route traffic.

In the study of networks a network's \emph{topology} refers to the network's structure of interactions while a network's \emph{dynamics} is the pattern of behavior exhibited by the network's elements \cite{BS09}. It is worth emphasizing that real-world networks are not only dynamic in terms of the behavior of their elements but also in terms of their topology, both of which effect the network's function. For example, the World Wide Web has an ever changing structure of interactions as web pages and the hyperlinks connecting these pages are updated, added, and deleted (see \cite{GS09} for a review of the evolving topology of networks).

A number of network formation models have been proposed to describe the type of growth observed in real networks. These models fall to a large extent into three categories. The first are those related to the Barabasi and Albert model \cite{Bara99} and its predecessor the Price model \cite{Price76}, in which elements are added one by one to a network and are preferentially attached to vertices with high \emph{degree}, i.e. to vertices with a high number of neighbors or some variant of this process \cite{Bara00,Doro00,Krap01}. The second are \emph{vertex copying models} in which new elements are added to a network by randomly choosing an existing network element. A new element is then connected to each neighbor of this existing network element with probability $p\in(0,1)$ or to some other element with probability $1-p$ \cite{Klein99,Sole02,Vaz03}. Third are the \emph{network optimization models} where the topology of the network is evolved to minimize some global constraint, e.g. operating costs vs. travel times in a transportation network \cite{Ferrer03,Gastner06}.

These models are devised to create networks that exhibit some of the most widely observed features found in real networks. This includes (i) having a degree distribution that follows a power-law, i.e. being \emph{scale-free}, (ii) having a \emph{disassortative neighbor} property where vertices with high (low) degree have neighbors with low (high) degree, (iii) having a high \emph{clustering coefficient} indicating the presence of many triangles within the network, and (iv) having the \emph{small-world property}, meaning that the average distance between any two network elements is logarithmic in the size of the network (see \cite{Newman10} for more details on these properties).

Aside from these structural features, one of the hallmarks of a real network is that, as its topology evolves, so does its ability to perform complex tasks. This happens, for instance, in neural networks, which become more modular in structure as individual parts of the brain become increasingly specialized in function \cite{Sporns13}. Similarly, gene regulatory networks can specialize the activity of existing genes to create new patterns of gene behavior \cite{Esp10}. In technological networks such as the internet, this differentiation of function is also observed and is driven by the need to handle and more efficiently process an increasing amount of information.

Here we propose a very different class of models than those described above, which are built on this notion of specialization. We refer to these as \emph{specialization models} of network growth. These models are based on the fundamental idea that as a network specializes the function of one or more of its \emph{components}, i.e. a subnetwork(s) that performs a specific function, the network first creates a number of copies of this component. These copies are attached to the network in ways that reflect the original connections the component had within the network only sparser. The new copies ``specialize" the function of the original component in that they carry out only those functions requiring these specific connections (cf. Figure \ref{fig00}).

The components which are specialized in this growth process form \emph{motifs}, i.e. statistically significant structures in which particular network functions are carried out. As copies of these motifs are placed throughout the network via the process of specialization the result is an increase in the network's modularity \cite{Milo02}. Consequently, repeated application of this process results in a \emph{hierarchical topology} in which modular structures appear at multiple scales. Moreover, because new components are far less connected to the network than the original components the result is an increasingly sparse network topology. Hence the network acquires a modular \cite{Newman2006}, hierarchical \cite{Clauset08,Leskovec2008}, as well as sparse topology \cite{N03,HG08} each of which is a distinguishing feature of real networks when compared, for instance, to random graphs (see \cite{GS09} Section 6.3.2.1).

Importantly, our model of network growth is extremely flexible in that a network can be uniquely specialized over any subset $B$ of its elements. We refer to any such subset $B$ as a network \emph{base}. Since $B$ can be any subset of a network's elements there is an significant number of ways in which a network can be specialized. Hence, \emph{many} network growth models. An obvious application is, given a particular network, to find a base over which this network can be specialized that evolves the its topology in a way that models its observed growth. Finding a \emph{rule} $\tau$ that generates this base is a natural objective of a network scientist using this model to investigate their particular network(s) of interests. The reason is that finding such a rule suggests a mechanism for the particular network's growth that can then be tested against the growth of the actual network.

A particularly simple rule we consider here is the rule $r=r_p$ for $p\in(0,1)$ that uniformly selects a random network base consisting of $p$ percent of the network's elements. Under this specialization rule we find that an initial network evolves under a sequence of specializations into a network that is numerically (i) \emph{scale free}, (ii) has a \emph{disassortative neighbor} property, (iii) has a power-law like \emph{clustering coefficients} and (iv) the \emph{small-world property}. Hence, this random variant of the specialization model appears to capture a number of the well-know properties observed in real networks. To our knowledge this is the only such class of models to capture these properties along with creating an increasingly sparse, modular, and hierarchical network topology.

Additionally, we show how specialization rules can be used to compare the topology of different networks. Specifically, two graphs $G$ and $H$ are considered to be \emph{similar} to each other with respect to a rule $\tau$ if they specialize to the same graph under $\tau$. This notion of similarity, which we refer to as \emph{specialization equivalence}, can be used to partition any set of networks into those that are similar, i.e. are specialization equivalent, with respect to a given rule $\tau$ and those that are not (see Section \ref{sec2}, Theorem \ref{thm2}). One reason for designing such a rule $\tau$ is that typically it is not obvious that two different graphs are in some sense equivalent. That is, two networks may be similar but until both are specialized with respect to $\tau$ this may be difficult to see. Here we show that by choosing an appropriate rule $\tau$ one can discover this similarity (see example \ref{ex:evoequ}). Of course, it is important that this rule be designed by the particular biologist, chemist, physicist, etc. to have some significance with respect to the nature of the networks under consideration.

In a following paper we rigorously show that these specialization models also preserve a number of spectral and dynamical properties of a network. Specifically, we show that as a network is specialized the eigenvalues of the resulting network are those of the original network together with the eigenvalues of the specialized components, extending the theory of isospectral network transformations found in \cite{BWBook}. Additionally, we show that the eigenvector centrality of the base vertices of a network remain unchanged as the network is specialized. In terms of dynamic properties we prove that if a dynamical network is intrinsically stable, which is a stronger form of a standard notion of stability (see \cite{BW13} for more details), then any specialized version of this network will also be intrinsically stable. Hence, network growth given by our models of specialization will not destabilize the network's dynamics if the network has this stronger version of stability. This is important in many real-world application since network growth can have a destabilizing effect on a network's dynamics, e.g. cancer, which is the abnormal growth of cells has this effect in biological networks.

The paper is organized as follows. In Section \ref{sec2} we introduce the specialization model and the notion of specialization rules. In Section \ref{sec3} we show that if a network is randomly specialized over a fixed percentage of its elements the result is a network that has many properties found in real-world networks including the small-world property, power-law like degree distributions and clustering coefficients. In Section \ref{sec4} we describe the notion of specialization equivalence and how this notion can be used to compare the structure of different networks. The last section, Section \ref{conc} contains some concluding remarks.

\section{Specialization Model of Network Growth}\label{sec2}

The standard method used to describe the topology of a network is a graph. A \emph{graph} $G=(V,E,\omega)$ is composed of a \emph{vertex set} $V$, an \emph{edge set} $E$, and a function $\omega$ used to weight the edges of the graph. The vertex set $V$ represents the \emph{elements} of the network, while the edges $E$ represent the links or \emph{interactions} between these network elements. The weights of the edges, given by $\omega$, give some measure of the \emph{strength} of these interactions. Here we consider weights that are real numbers, which account for the vast majority of weights used in network analysis \cite{Newman10}.

The edges $E$ of a graph can either be \emph{directed} or \emph{undirected}, \emph{weighted} or \emph{unweighted}. Here we consider, without loss in generality, those graphs that have weighted directed edges. The reason is that an undirect edge is equivalent to two directed edges pointing in either direction and any unweighted edge can be weighted by giving the edge unit weight.

For the graph $G=(V,E,\omega)$ we let $V=\{v_1,\dots,v_n\}$, where $v_i$ represents the $i$th network element. We let $e_{ij}$ denote the directed edge that begins at $v_i$ and ends at $v_j$. In terms of the network, the edge $e_{ij}$ belongs to the edge set $E$ if the $i$th network element has some direct influence or is linked to the $j$th network element.

As mentioned in the introduction, one of the hallmarks of real networks is that as a network evolves so does its ability to perform various tasks. It has been observed that to accomplish this a network will often specialize the tasks performed by of one or more of its components, i.e. subnetworks. As motivation for our model of network growth we give the following example of network specialization.

\begin{figure}
\vspace{0.75cm}
\begin{center}
\begin{tabular}{cc}
    \begin{overpic}[scale=.33]{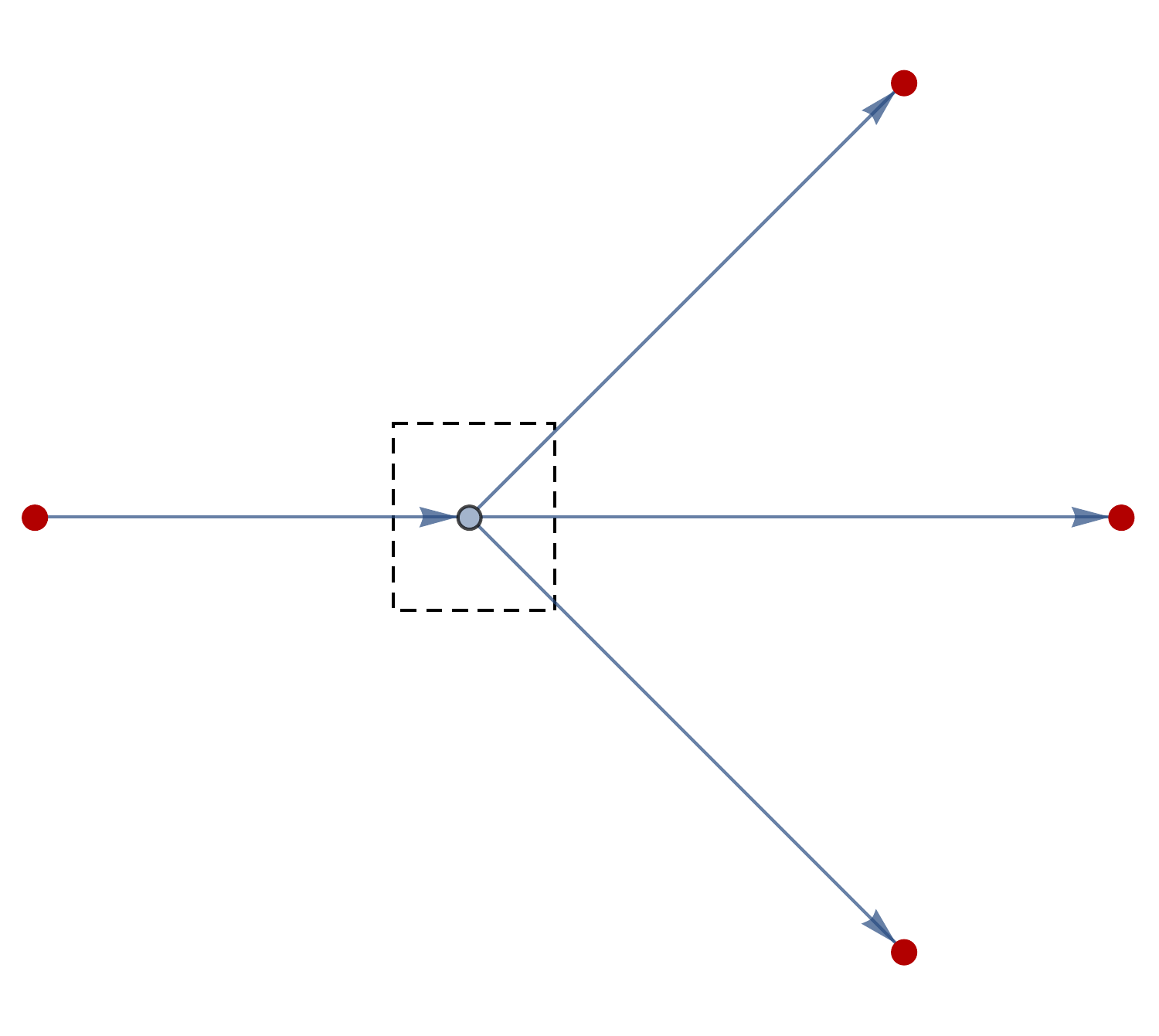}
    \put(75,12){\scriptsize{related ``elements"}}
    \put(87,6){\scriptsize{pages}}
    \put(88,49){\scriptsize{related ``planets"}}
    \put(101,43){\scriptsize{pages}}
    \put(75,86){\scriptsize{related ``mythology"}}
    \put(87,80){\scriptsize{pages}}
    \put(27,60){\scriptsize{``Mercury"}}
    \put(35,55){\scriptsize{page}}
    \put(15,-15){Undifferentiated ``Mercury"}
    \put(33,-24){Wikipedia Page}
    \end{overpic} &
    \hspace{2.25cm}
    \begin{overpic}[scale=.33]{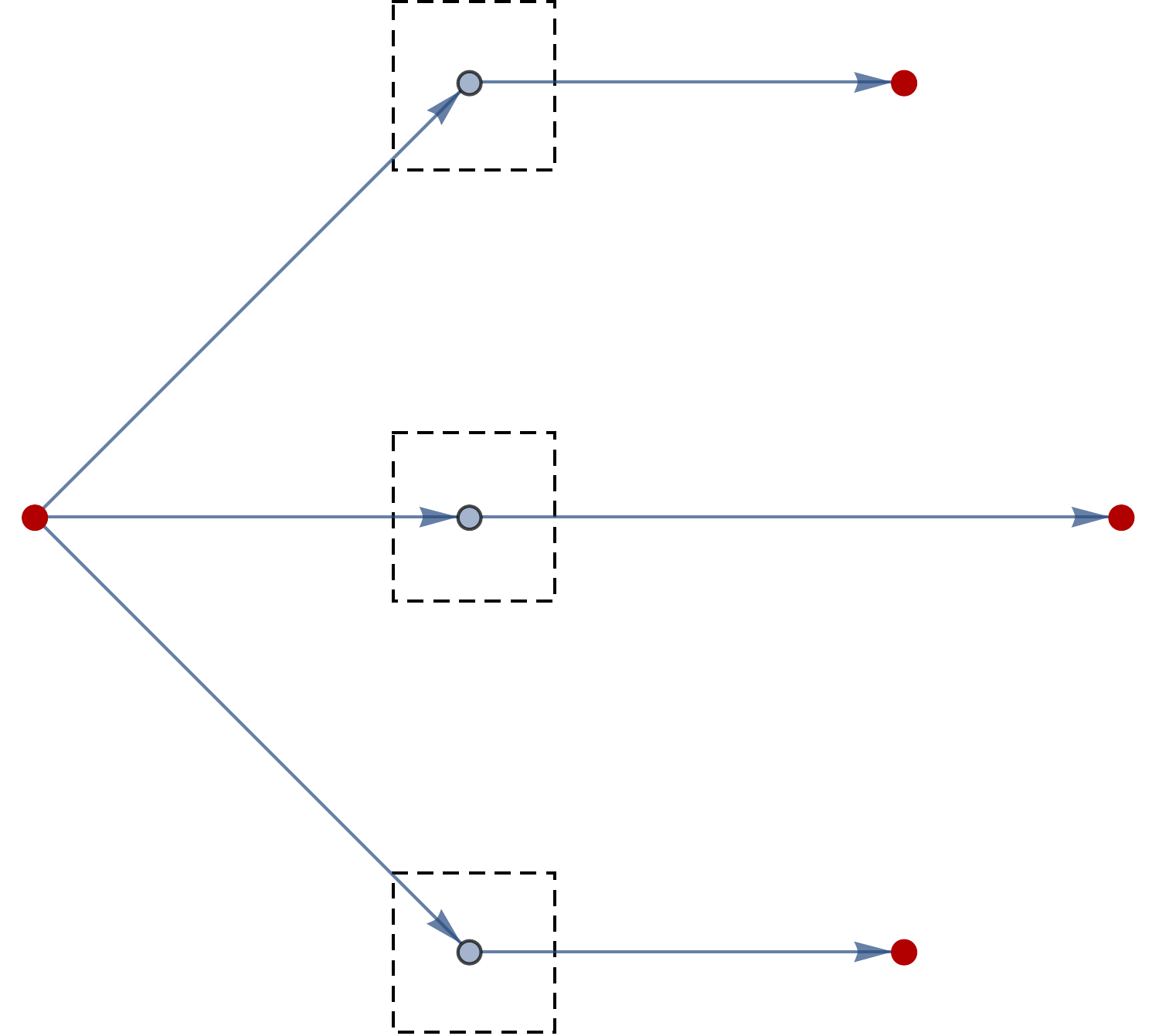}
    \put(69,12){\scriptsize{related ``elements"}}
    \put(81,6){\scriptsize{pages}}
    \put(84,49){\scriptsize{related ``planets"}}
    \put(100,43){\scriptsize{pages}}
    \put(69,86){\scriptsize{related ``mythology"}}
    \put(81,80){\scriptsize{pages}}
    \put(26,62){\scriptsize{``Mercury"}}
    \put(31,55){\scriptsize{(planet)}}
    \put(26,100){\scriptsize{``Mercury"}}
    \put(25,93){\scriptsize{(mythology)}}
    \put(27,24){\scriptsize{``Mercury"}}
    \put(29,17){\scriptsize{(element)}}
    \put(15,-15){Disambiguated ``Mercury"}
    \put(30,-24){Wikipedia Pages}
    \end{overpic}
\end{tabular}
\vspace{1cm}
\end{center}
  \caption{This figure shows the effect of disambiguating the Wikipedia page on ``Mercury" into three distinct webpages, which are respectively Mercury the mythological figure, Mercury the planet, and mercury the element.}\label{fig00}
\end{figure}
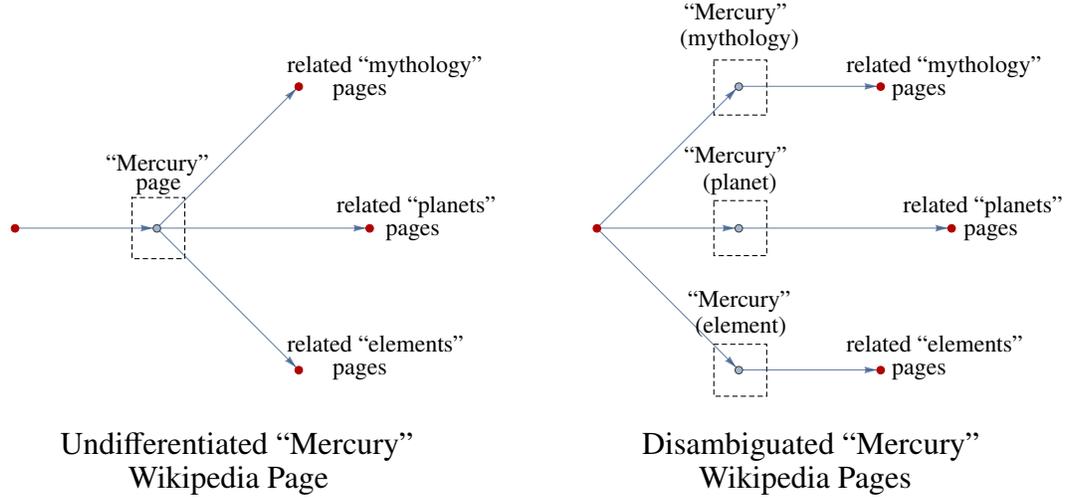

\begin{example}\label{ex:dis}\textbf{(Wikipedia Disambiguation)}
The website Wikipedia is a collection of webpages consisting of articles that are linked by topic. The website evolves as new articles are either added, linked, and modified within the existing website. One of the ways articles are added, linked, and modified is that some article within the website is disambiguated. If an article's content is deemed to refer to a number of distinct topics then the article can be \emph{disambiguated} by separating the article into a number of new articles, each on a more specific or \emph{specialized} topic.

Wikipedia's own page on disambiguation gives the example that the word ``Mercury" can refer to either Mercury the \emph{mythological figure}, Mercury the \emph{planet}, or mercury the \emph{element} \cite{Wiki17}. To emphasize these differences the Wikipedia page on Mercury has been disambiguated into three pages on Mercury; one for each of these subcategories. Users arriving at the Wikipedia ``Mercury" page \cite{Merc17} are redirected to these pages (among a number of other related pages).

The result of this disambiguation is shown in Figure \ref{fig00}. In the original undifferentiated Mercury page users arriving from other pages could presumably find links to other ``mythology", ``planet", and ``element"  pages (see Figure \ref{fig00}, left). After the page was disambiguated users were linked to the same pages but only those relevant to the particular ``Mercury" page they had chosen (see Figure \ref{fig00}, right). In terms of the topology of the network, this disambiguation results in the creation of a number of new ``Mercury" pages each of which is linked to a subset of pages that were linked to the original ``Mercury" page. Growth via disambiguation is a result of the new ``copies" of the original webpage.

However, what is important to the functionality of the new specialized network is that the way in which these new copies are linked to the unaltered pages should reflect the topology of the original network. In our model the way in which we link these new components, which can be much more complex than single vertices, is by separating out the paths and cycles on which these components lie, in a way that mimics the original network structure.
\end{example}

Hence, to describe our models of network specialization and its consequences we first need to consider the paths and cycles of a graph. A \emph{path}\index{path} $P$ in the graph $G=(V,E,\omega)$ is an ordered sequence of distinct vertices $P=v_1,\dots,v_m$ in $V$ such that $e_{i,i+1}\in E$ for $i=1,\dots,m-1$. If the vertices $v_1$ and $v_m$ are the same then $P$ is a \emph{cycle}\index{cycle}. If it is the case that a cycle contains a single vertex then we call this cycle a \emph{loop}\index{loop}.

Another fundamental concept that is needed is the notion of a strongly connected component. A graph $G=(V,E,\omega)$ is \emph{strongly connected} if for any pair of vertices $v_i,v_j\in V$ there is a path from $v_i$ to $v_j$ or $G$ consists of a single vertex. A \emph{strongly connected component} of a graph $G$ is a subgraph that is strongly connected and is maximal with respect to this property.

Because we are concerned with evolving the topology of a network in ways that preserve, at least locally, the network's topology we will also need the notion of a graph restriction. For a graph $G=(V,E,\omega)$ and a subset $B\subseteq V$ we let $G|_{B}$ denote the \emph{restriction} of the graph $G$ to the vertex set $B$, which is the subgraph of $G$ on the vertex set $B$ along with any edges of the graph $G$ between vertices in $B$. We let $\bar{B}$ denote the \emph{complement} of $B$, so that the restriction $G|_{\bar{B}}$ is the graph restricted to the complement of those vertices not in $B$.

The key to specializing the structure of a graph is to look at the strongly connected components of the restricted graph $G|_{\bar{B}}$. If $S_1,\dots,S_m$ denote these strongly connected components then we will need to find paths or cycles of these components, which we refer to as \emph{components branches}.

\begin{figure}
\begin{center}
\begin{tabular}{c}
    \begin{overpic}[scale=.33]{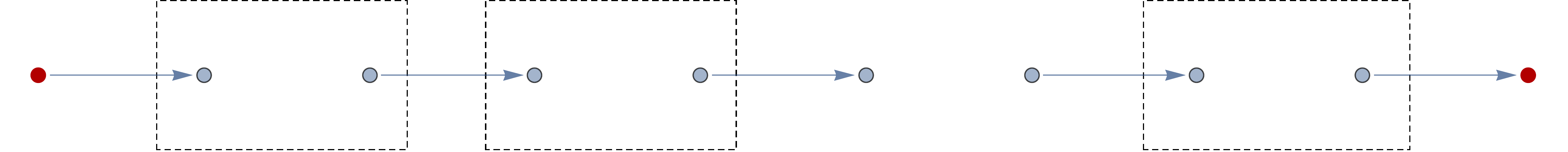}
    \put(-1,4.5){$v_i$}
    \put(6,6){$e_0$}
    \put(16,4){$S_1$}
    \put(27,6){$e_1$}
    \put(38,4){$S_2$}
    \put(49,6){$e_2$}
    \put(56.75,4.55){\Huge$\dots$}
    \put(66,6){$e_{m-1}$}
    \put(79.5,4){$S_m$}
    \put(91,6){$e_m$}
    \put(99,4.5){$v_j$}
    \end{overpic}
\end{tabular}
\end{center}
  \caption{A representation of a path of components is shown, consisting of the sequence $S_1,\dots,S_m$ of components beginning at vertex $v_i\in B$ and ending at vertex $v_j\in B$. From $S_k$ to $S_{k+1}$ there is a single directed edge $e_{k+1}$. From $v_i$ to $S_1$ and from $S_m$ to $v_j$ there is also a single directed edge.}\label{fig01}
\end{figure}
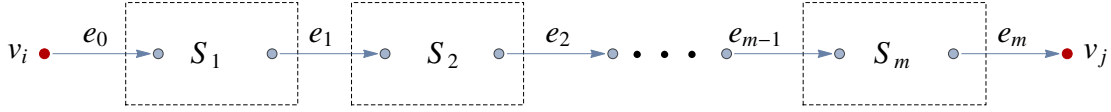

\begin{definition}\label{def:componentbranch} \textbf{(Component Branches)}
For a graph $G=(V,E,\omega)$ and vertex set $B\subseteq V$ let $S_1,\dots,S_m$ be the strongly connected components of $G|_{\bar{B}}$. If there are edges $e_0,e_1,\dots,e_m\in E$ and vertices $v_i,v_j\in B$ such that\\
(i) $e_k$ is an edge from a vertex in $S_k$ to a vertex in $S_{k+1}$ for $k=1,\dots,m-1$;\\
(ii) $e_0$ is an edge from $v_i$ to a vertex in $S_1$; and\\
(iii) $e_m$ is an edge from a vertex in $S_m$ to $v_j$, then we call the ordered set
\[
\beta=\{v_i,e_{0},S_1,e_{1},S_2,\dots,S_m,e_{m},v_{j}\}
\]
a \emph{path of components} in $G$ with respect to $B$. If $v_i=v_j$ then $\beta$ is a \emph{cycle of components}. We call the collection $\mathcal{B}_B(G)$ of these paths and cycles the \emph{component branches} of $G$ with respect to the base set of vertices $B$.
\end{definition}

A representation of the path of components is shown in Figure \ref{fig01}. The sequence of components $S_1,\dots,S_m$ in this definition can be empty in which case $m=0$ and $\beta$ is the path $\beta=\{v_i,v_j\}$ or loop if $v_i=v_j$. It is worth emphasizing that each branch $\beta\in\mathcal{B}_B(G)$ is a subgraph of $G$. Consequently, the edges of $\beta$ inherit the weights they had in $G$ if $G$ is weighted. If $G$ is unweighted then its component branches are likewise unweighted.

Once a graph has been decomposed into its various branches we construct the specialized version of the graph by merging these branches as follows.

\begin{definition} \textbf{(Graph Specialization)}\label{def:exp}
Suppose $G=(V,E,\omega)$ and $B\subseteq V$. Let $\mathcal{S}_B(G)$ be the graph which consists of the component branches $\mathcal{B}_{B}(G)=\{\beta_1,\dots,\beta_{\ell}\}$ in which we \emph{merge}, i.e. identify, each vertex $v_i\in B$ in any branch $\beta_j$ with the same vertex $v_i$ in any other branch $\beta_k$. We refer to the graph  $\mathcal{S}_B(G)$ as the \emph{specialization} of $G$ over the \emph{base} vertex set $B$.
\end{definition}

A specialization of a graph $G$ over a base vertex set $B$ is a two step process. The first step is the construction of the components branches. The second step is the merging of these components into a single graph. We note that, in a component branch $\beta\in\mathcal{B}_B(G)$ only the first and last vertices of $\beta$ belong to the base $B$. The specialized graph $\mathcal{S}_B(G)$ is therefore the collection of branches $\mathcal{B}_B(G)$ in which we identify an endpoint of two branches if they are the same vertex. This is demonstrated in the following example.

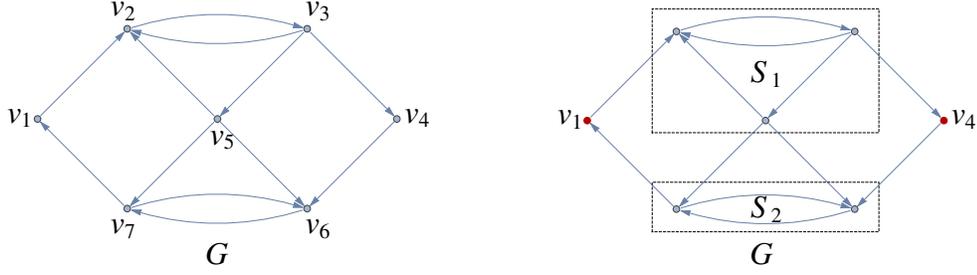
\begin{figure}
\begin{center}
\begin{tabular}{cc}
    \begin{overpic}[scale=.2]{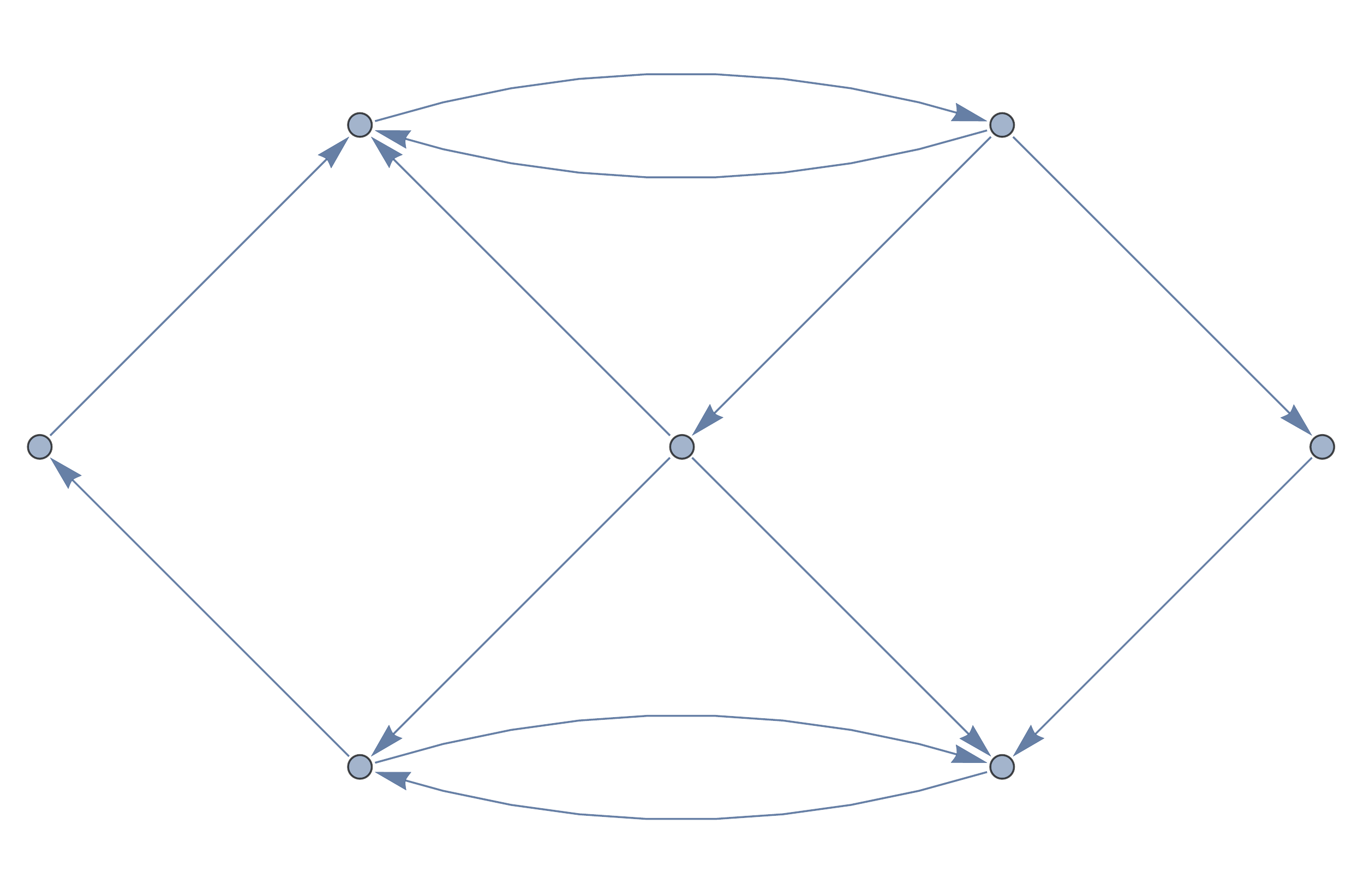}
    \put(47,-5){$G$}
    \put(-5,32){$v_1$}
    \put(22,60){$v_2$}
    \put(73,60){$v_3$}
    \put(99,32){$v_4$}
    \put(48,26.3){$v_5$}
    \put(73,3){$v_6$}
    \put(22,3){$v_7$}
    \end{overpic} &
    \hspace{0.7in}
    \begin{overpic}[scale=.2]{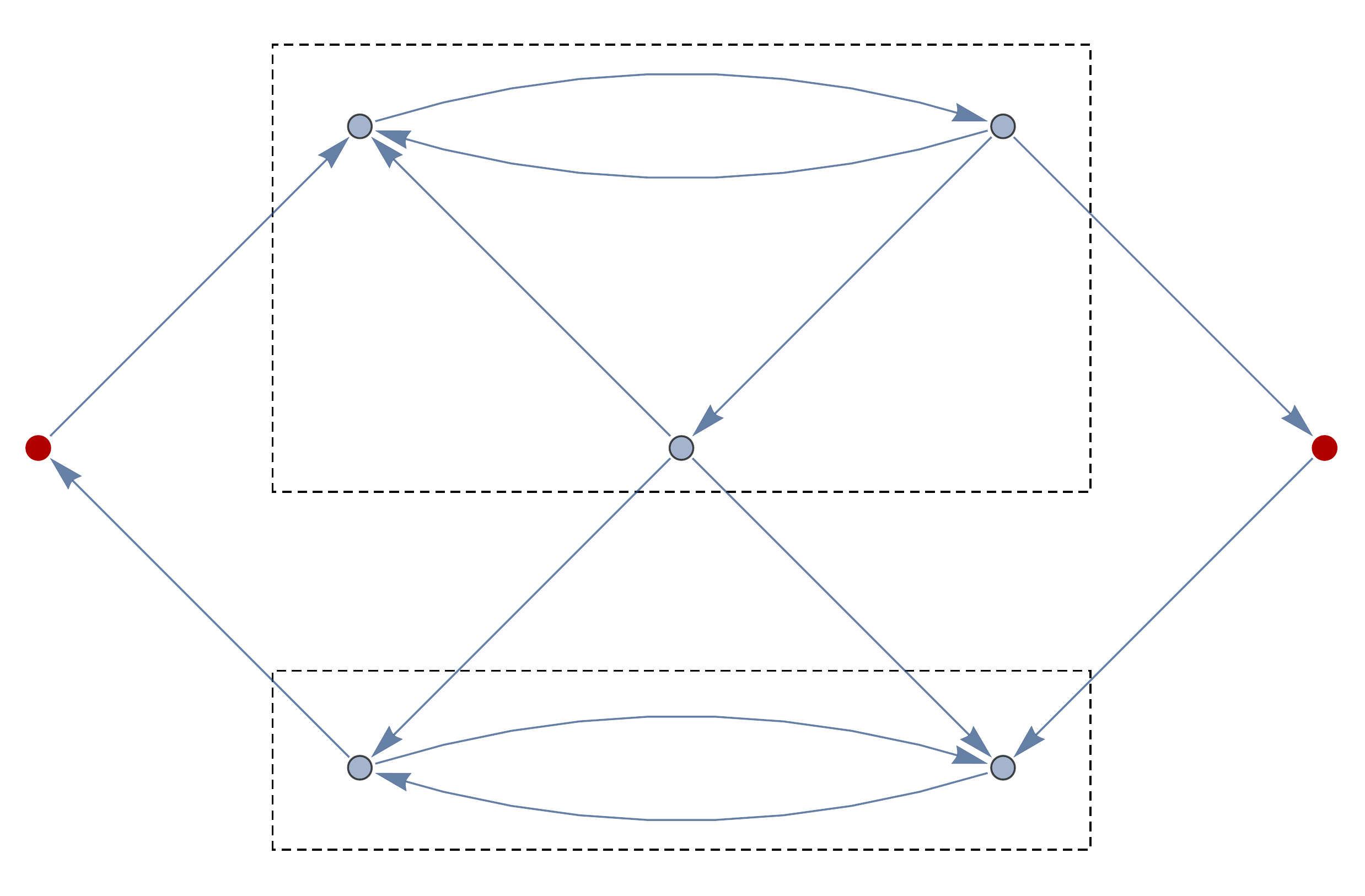}
    \put(46,-5){$G$}
    \put(-5,32){$v_1$}
    \put(99,32){$v_4$}
    \put(46,43){$S_1$}
    \put(46,7.25){$S_2$}
    \end{overpic}
\end{tabular}
\end{center}
  \caption{The unweighted graph $G=(V,E)$ is shown left. The components $S_1$ and $S_2$ of $G$ with respect to the vertex set $B=\{v_1,v_4\}$ are shown right. These components are the subgraphs $S_1=G|_{\{v_2,v_3,v_5\}}$ and $S_2=G|_{\{v_6,v_7\}}$, indicated by the dashed boxes, which are the strongly connected components of the restricted graph $G|_{{\bar{B}}}$.}\label{fig1}
\end{figure}

\begin{example}\label{ex:2} \textbf{(Constructing Graph Specializations)}
Consider the \emph{unweighted} graph $G=(V,E)$ shown in Figure \ref{fig1} (left). For the base vertex set $B=\{v_1,v_4\}$ the specialization $\mathcal{S}_B(G)$ is constructed as follows.\\

\noindent\emph{Step 1:} \emph{Construct the branch components of $G$ with respect to $B$.} The graph $G|_{\bar{B}}$ has the strongly connected components $S_1=G|_{\{v_2,v_3,v_5\}}$ and $S_2=G|_{\{v_6,v_7\}}$, which are indicated in Figure \ref{fig1} (right). The set $\mathcal{B}_B(G)$ of all paths and cycles of components beginning and ending at vertices in $B$ consists of the component branches
\begin{align*}
\beta_1&=\{v_1,e_{12},S_1,e_{34},v_4\} \ \ \ \ \ \ \ \ \ \ \ \ \ \beta_2=\{v_4,e_{46},S_2,e_{71},v_1\}\\
\beta_3&=\{v_1,e_{12},S_1,e_{56},S_2,e_{71},v_1\} \ \ \ \beta_4=\{v_1,e_{12},S_1,e_{57},S_2,e_{71},v_1\};
\end{align*}
which are shown in Figure \ref{fig2} (left).\\

\noindent\emph{Step 2:} \emph{Merging the branch components.} By merging each of the vertices $v_1\in B$ in all branches of $\mathcal{B}_B(G)=\{\beta_1,\beta_2,\beta_3,\beta_4\}$ shown in Figure \ref{fig2} (left) and doing the same for the vertex $v_4\in B$, the result is the graph $\mathcal{S}_B(G)$ shown in Figure \ref{fig2} (right), which is the specialization of $G$ over the base vertex subset $B$.
\end{example}

\begin{figure}
\begin{center}
\begin{tabular}{cc}
    \begin{overpic}[scale=.55]{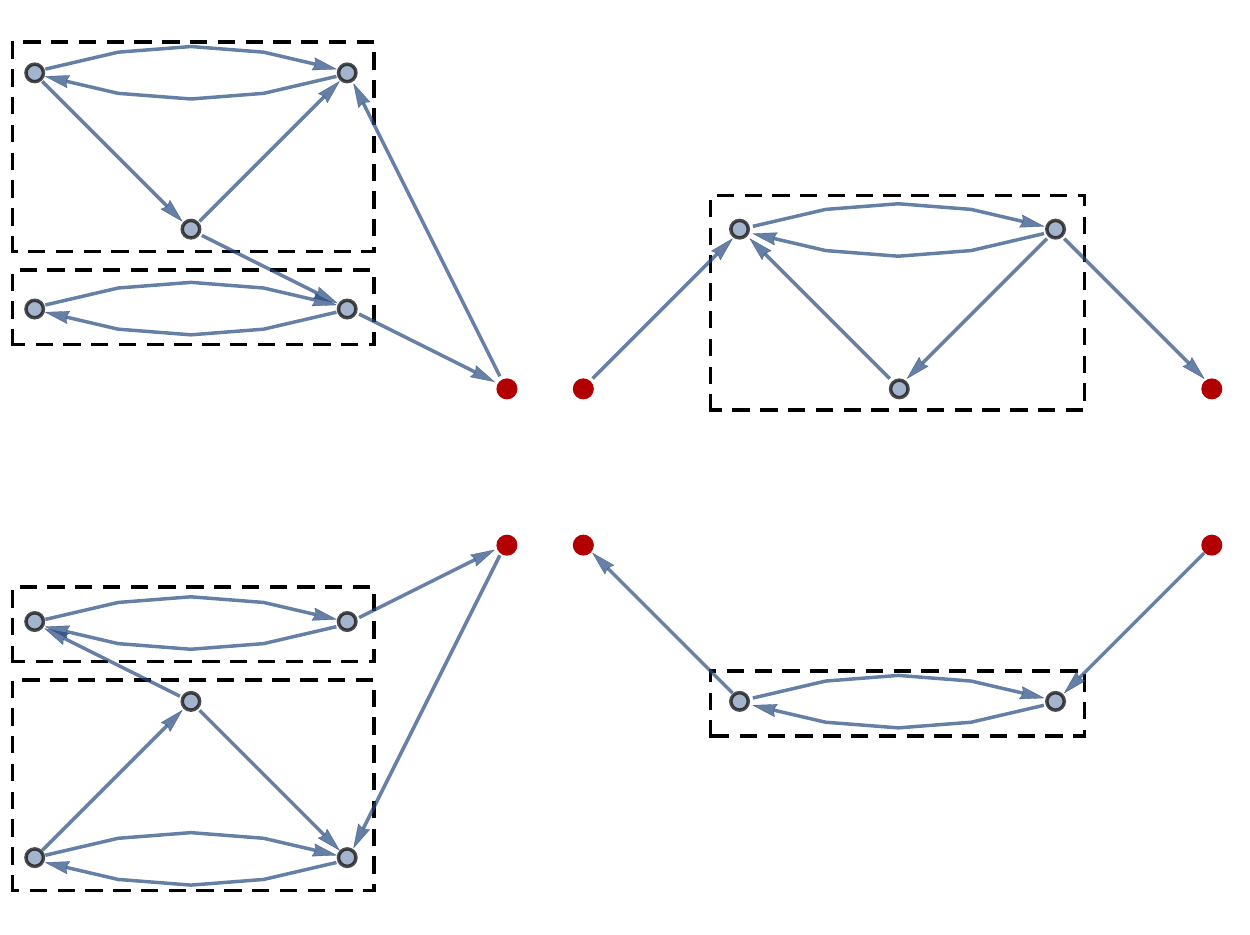}
    \put(49,0){$\mathcal{B}_B(G)$}
    \put(14,74){$\beta_4$}
    \put(14,-2){$\beta_3$}
    \put(70,10){$\beta_2$}
    \put(70,62){$\beta_1$}

    \put(94.5,39.5){$v_4$}
    \put(95,33){$v_4$}
    \put(38,39.5){$v_1$}
    \put(38,33){$v_1$}

    \put(45,39.5){$v_1$}
    \put(45,33){$v_1$}
    \end{overpic} &
    \hspace{0.55cm}
    \begin{overpic}[scale=.48]{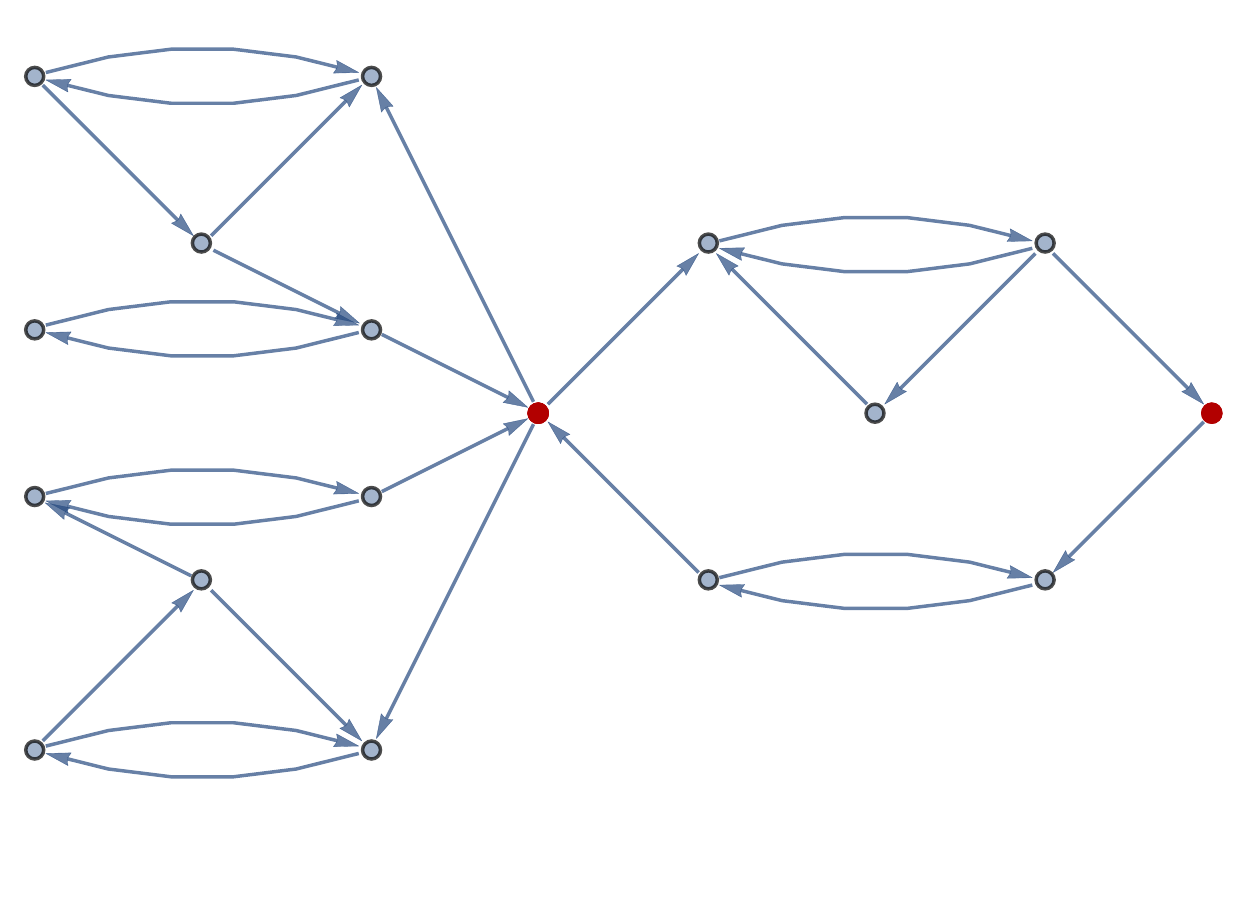}
    \put(45,0){$\mathcal{S}_B(G)$}

    \put(29,39){$v_1$}
    \put(100,39){$v_4$}
    \end{overpic}
\end{tabular}
\end{center}
  \caption{The component branches $\mathcal{B}_B(G)=\{\beta_1,\beta_2,\beta_3,\beta_4\}$ of the graph $G=(V,E)$ from figure \ref{fig1} over the base vertex set $B=\{v_1,v_4\}$ are shown (left). The specialized graph $\mathcal{S}_B(G)$ is shown (right), which is made by merging each of the vertices $v_1$ and $v_4$ respectively in each of the branches of $\mathcal{B}_B(G)$. The edge labels and vertex labels are omitted, except for those vertices in $B$, to emphasize how these vertices are identified.}\label{fig2}
\end{figure}
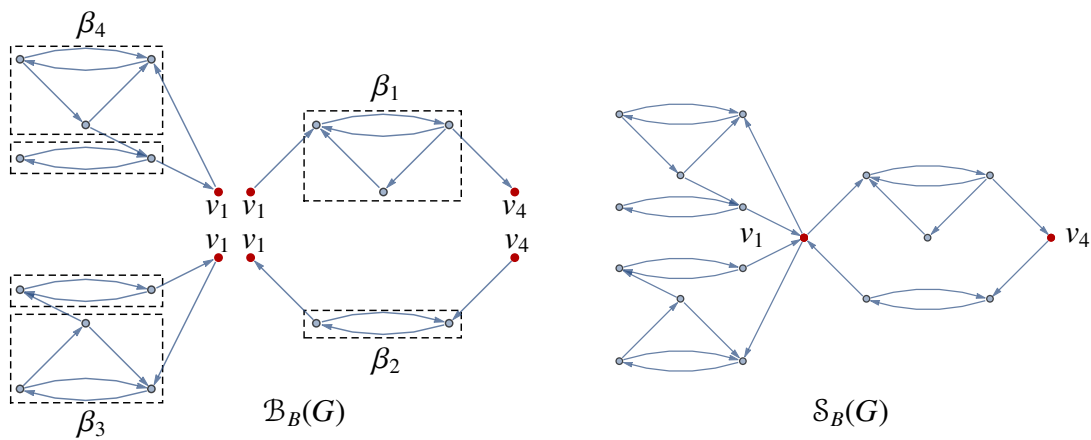

To summarize, our model of network growth consists in evolving the topology of a given network by selecting some base subset of the network's elements and specializing the graph associated with the network over the corresponding vertices. We refer to this process as the \emph{specialization model} of network growth.

The idea is that in an information network, such as the World Wide Web, this model of specialization can model the differentiation and inclusion of new information (cf. Example \ref{ex:dis}). In a biological model this can be used to describe various developmental processes, for instance, the specialization of cells in to tissue. In a social network this model can similarly be used to model how new relationships are formed when, for example, an individual is introduced by someone to their immediate group of friends.

It is also worth mentioning that a network specialization evolves the network's topology by maintaining the interactions between its base elements $B$ and by differentiating the other network functions into sequences of components. The result is a network with many more of these components. These components are important for a number reasons. The first is that they form network \emph{motifs}, which are statistically important subgraphs within the network that typically perform a specific network function \cite{A07}. Second, because there are very few connections between these components the resulting network has a far more modular structure, which is a feature found in real networks (see \cite{Newman2006} for a survey of modularity). Third, because of the number of copies of the same component, the specialized graph has a certain amount of redundancy in its topology, which is another feature observed in real networks \cite{MSA08,TSE99}. Last, specializations results in \emph{sparser} graphs, i.e. graphs in which the ratio of edges to vertices is relatively small, which is again a characteristic found in real networks \cite{N03,HG08}.

Additionally, many networks exhibit hierarchical organization, in which network vertices divide into groups or components that further subdivide into smaller groups of components, and so on over multiple scales. It has been observed that these components often come from the same functional units, e.g. ecological niches in food webs, modules in biochemical networks including protein interaction networks, metabolic networks or genetic regulatory networks or communities in social networks \cite{Clauset08,Leskovec2008}. Because new components are created each time a graph is specialized a network becomes increasingly hierarchial as this process of specialization is repeated.

\section{Specialization Rules}\label{sec3}

As a significantly large number of bases are possible for any reasonably sized network a natural question is, given a particular real-world network (or class of networks) can we find a base or sequence of bases that can be used to model this network's growth via specialization. Another way of stating this is, is it possible to find a specialization rule that selects network base(s) that can be used to specialize the network in a way that mimics its actual growth.

Here a \emph{specialization rule} $\tau$ is a rule that selects for any graph $G=(V,E,\omega)$ a base vertex subset $V_{\tau}(G)\subseteq V$. For simplicity, we let
\[
\tau(G)=S_{V_{\tau}(G)}(G) \ \ \text{and} \ \ \tau^k(G)=\tau(\tau^{k-1}(G)) \ \ \text{for} \ \ k>0
\]
denote the specialization of $G$ with respect to $\tau$ and the $k$th specialization of $G$ with respect to $\tau$, respectively. Each rule $\tau$ generates a different type of growth and as such can be thought of as inducing a different model of network growth.

The specialization $\tau(G)$ is \emph{unique} if $\tau$ selects a unique base vertex subset of $G$. Not all rules produce a unique outcome as $\tau$ can be a rule that selects vertices of $G$ in some random way. For the moment we consider an important example of a random specialization rule. The reason we focus on this particular rule is that its repeated use leads to graphs (networks) that exhibit some of the most well-know properties observed in real networks.

\begin{example}\label{ex:rand}\textbf{(Random Specializations)}
For $p\in(0,1)$ let $r=r_p$ be the specialization rule that uniformly selects a random network base consisting of $p$ percent of the network's elements rounded to the nearest whole number. To understand the effect this rule has on general types of networks we start by first describing its effect on a single specific network.

Displayed in Figure \ref{fig:1} is an initial network given by the graph $G$ on ten vertices and fifteen edges, which is sequentially specialized using the rule $r=r_p$ for $p=.91$. The graph $G$ is specialized a total of seven times, which generates the sequence $G,r(G),\dots,r^{7}(G)$ of graphs. The top row of Figure \ref{fig:1} shows the graphs $G$, $r^2(G)$, $r^4(G)$, and $r^{6}(G)$.

\begin{figure}
\begin{center}
\begin{overpic}[scale=.39]{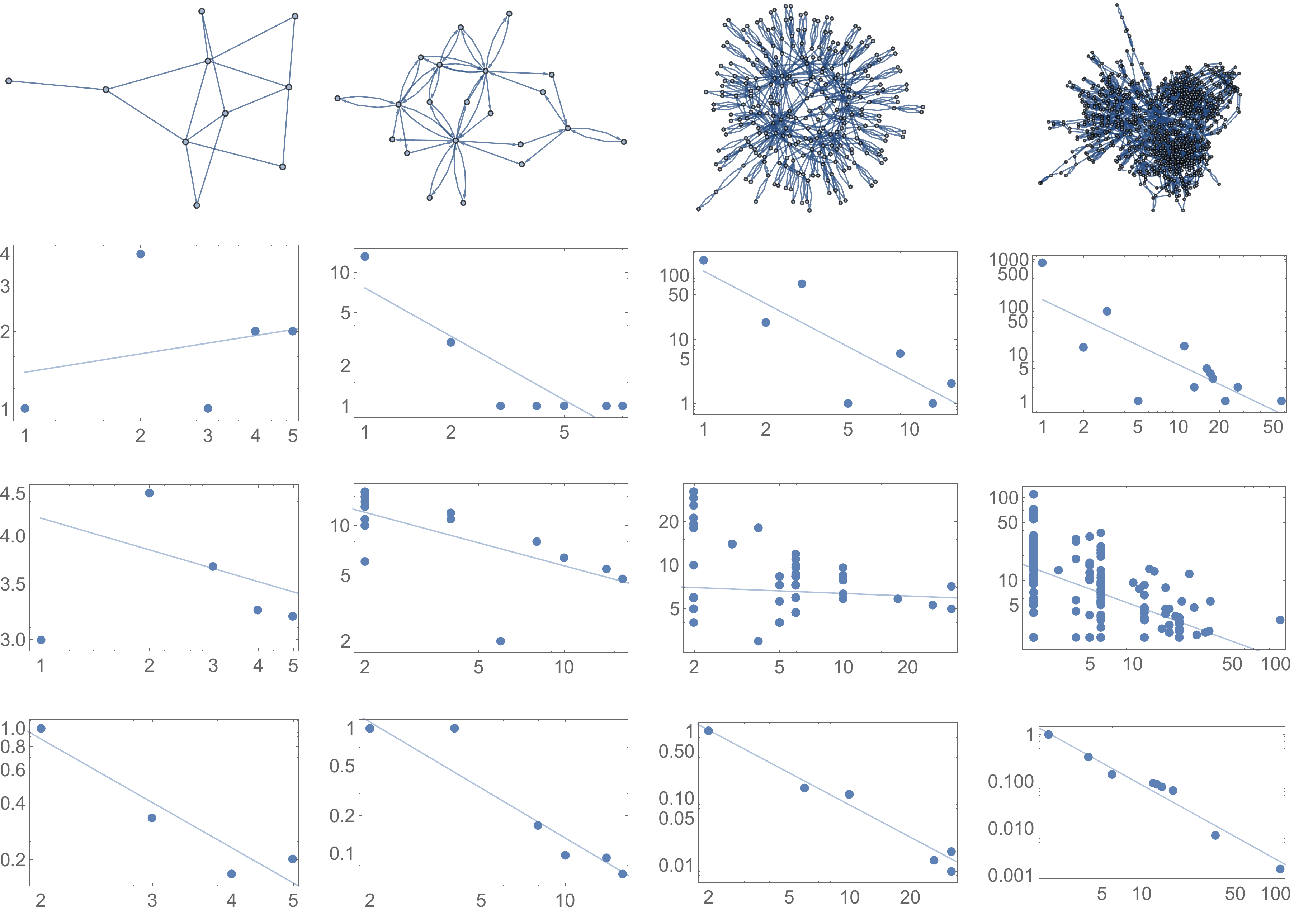}
\put(13,53.25){\fontsize{8}{5} $G$}
\put(34.5,53.25){\tiny$r^2(G)$}
\put(60,53.25){\tiny$r^4(G)$}
\put(86,53.25){\tiny$r^6(G)$}
\put(6,-1){\tiny vertex degree}
\put(32,-1){\tiny vertex degree}
\put(57,-1){\tiny vertex degree}
\put(83.5,-1){\tiny vertex degree}
\put(-2,43){\rotatebox{90}{ \fontsize{7}{5}\selectfont {count}}}
\put(-2,20){\rotatebox{90}{ \fontsize{7}{5}\selectfont {mean nbr. degree}}}
\put(-2,3){\rotatebox{90}{ \fontsize{7}{5}\selectfont {clustering coef.}}}
\put(11.8,50.1){\fontsize{8}{5} $\gamma_0=-0.24$}
\put(38,50.1){\fontsize{8}{5} $\gamma_2=1.20$}
\put(64,49.8){\fontsize{8}{5} $\gamma_4=1.69$}
\put(89,49.7){\fontsize{8}{5} $\gamma_6=1.37$}
\put(13,32){\fontsize{8}{5} $\eta_0=0.13$}
\put(38,32){\fontsize{8}{5} $\eta_2=0.46$}
\put(64,32){\fontsize{8}{5} $\eta_4=0.06$}
\put(89,32){\fontsize{8}{5} $\eta_6=0.63$}
\put(13,13.5){\fontsize{8}{5} $\alpha_0=1.93$}
\put(38,13.5){\fontsize{8}{5} $\alpha_2=1.33$}
\put(64,13.5){\fontsize{8}{5} $\alpha_4=1.60$}
\put(89,13.5){\fontsize{8}{5} $\alpha_6=1.58$}
\end{overpic}
\end{center}
  \caption{The top row shows the graph $G$ as well as its specializations $r^2(G)$, $r^4(G)$, and $r^{6}(G)$ for $r=r_{.91}$. Each column in this figure corresponds to the graph in the top row, respectively. The second row shows the degree distribution of each graph on a log-log plot with a best-fit line with slope $-\gamma_i$, demonstrating the degree to which each distribution follows a power-law. Row three displays each graph's assortativity (mean-neighbor degree distribution) plotted on a log-log plot with a best-fit line with slope $-\eta_i$, where an increasingly downward trend can be seen. Row four displays each graph's clustering coefficients plotted vs. vertex degree on a log-log plot. A best-fit line with slope $-\alpha_i$ is given for each demonstrating the degree to which this data follows a power-law.}\label{fig:1}
\end{figure}
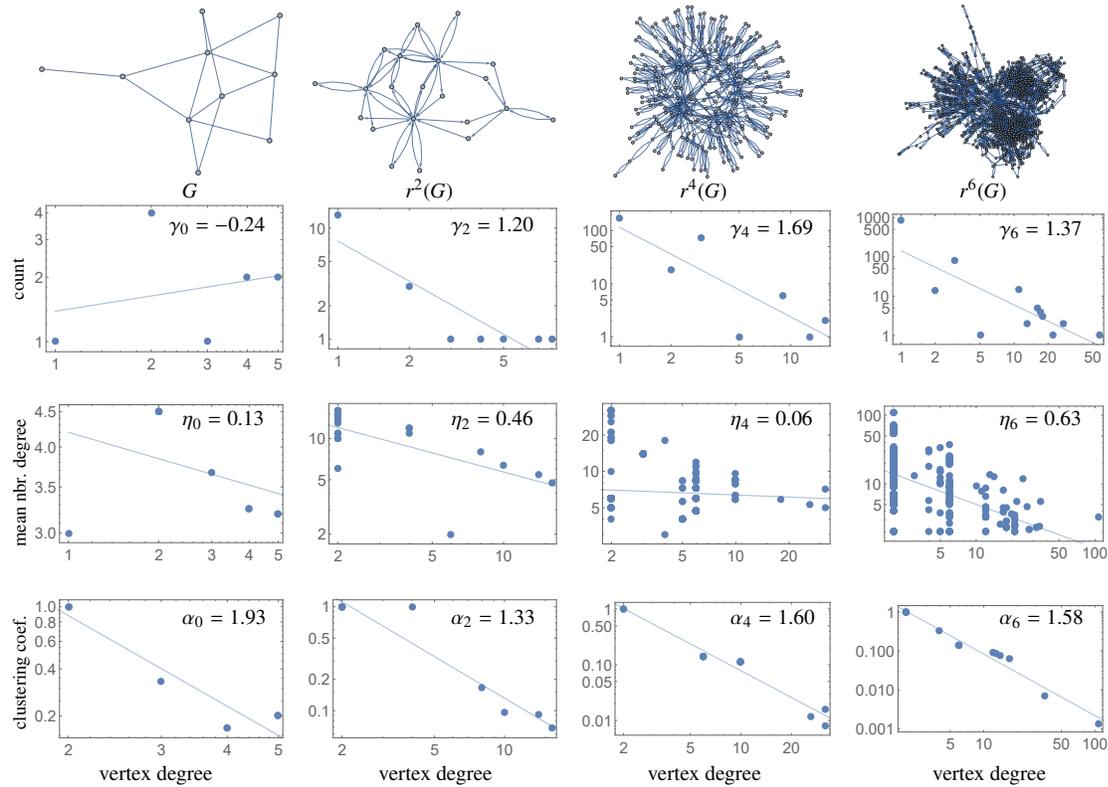

To get a sense of how this sequence of graphs compare to real observed networks, we investigate the (i) degree distribution, (ii) the assortativity, and (iii) the clustering coefficients of these graphs. The degree distribution of each graph is shown in the second row of Figure \ref{fig:1}, in which the number of vertices of each graph of a specific degree is plotted vs. vertex degree on a log-log plot. For each we plot a \emph{best-fit line}, which is the straight-line approximation of this data that minimizes the associated $R^2$-value, i.e. coefficient of determination, of the original data plotted on a linear-linear plot ignoring zeros. We note that as the graph is repeatedly specialized the $R^2$-value of these best-fit lines increase meaning that the graph's degree distribution becomes more power-law like as its topology is specialized.

The third row of the figure shows a plot of each graph's assortativity (mean-neighbor degree), which is plotted vs. the degree of each vertex. As the graph is specialized we see a strong downward trend, meaning that, on average, vertices with a large number of neighbors are increasingly connected to vertices that have a small number of neighbors. That is, the network becomes increasingly \emph{disassortative} as it is specialized.

The last row in Figure \ref{fig:1} displays the clustering coefficient for each graph plotted vs. vertex degree on a log-log plot. Again, for each we plot a \emph{best-fit line}. As can be seen in the figure the graph evolves under the rule $r=r_{.91}$ in such a way that its distribution of clustering coefficients becomes more and more colinear meaning that this distribution is becoming increasingly like a power-law.

For this sequence of graphs we also track how the mean distance between all pairs of vertices change with each iteration. This is shown in Figure \ref{fig:1a} where each point is an iterate's mean shortest-distance between all pair of vertices plotted vs. the number of vertices in the graph. These points are then fit with the logarithmic function $L(j)=\log_{\beta}(c j)$ with base $\beta\approx 2.35$ and some constant $c$. As this fit is quite good ($R^2=0.999598$), suggesting that this sequence of specializations has what is referred to as the \emph{small-world property} (or small-world effect).

\begin{figure}
\begin{center}
\begin{overpic}[scale=.25]{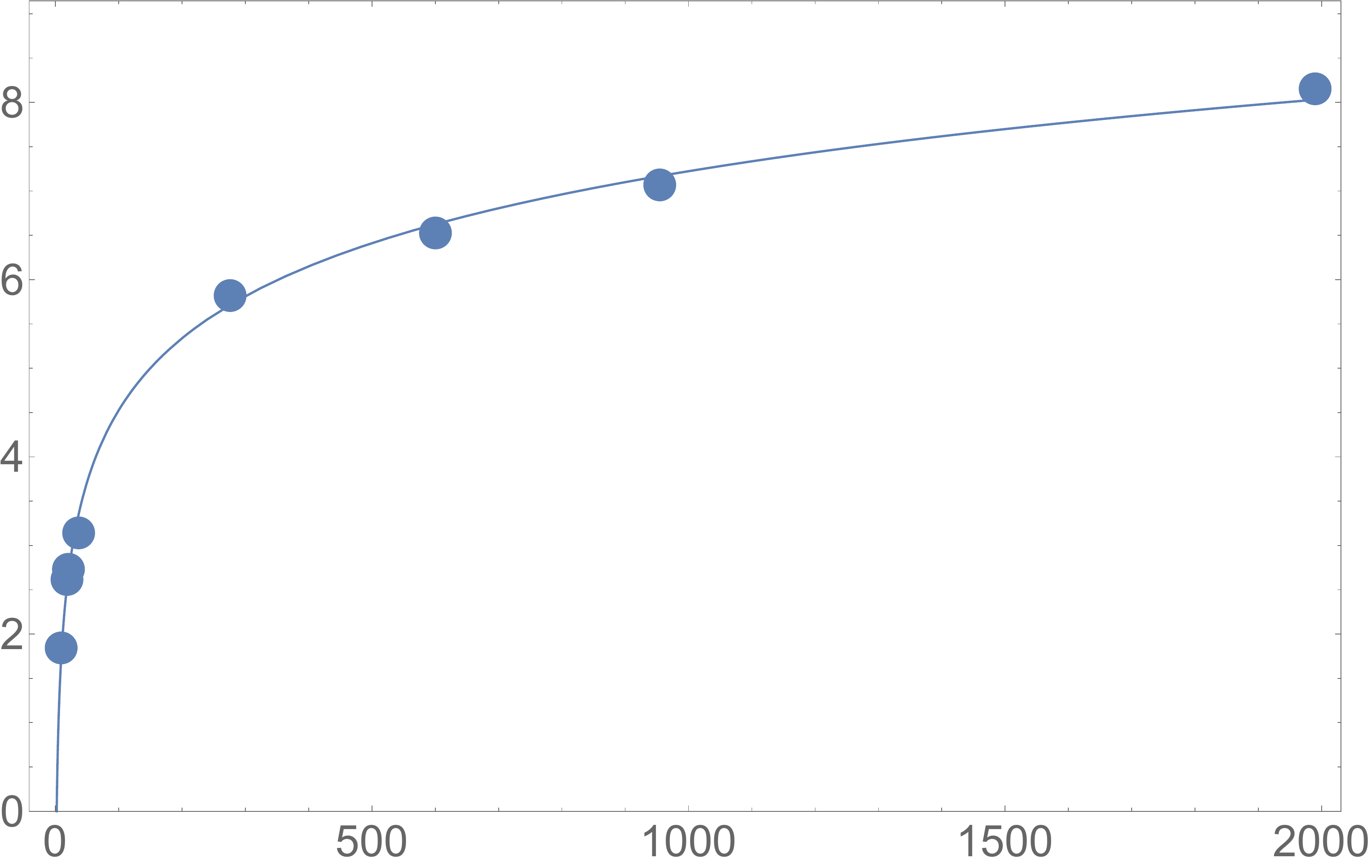}
\put(35,-3){ \fontsize{9}{5}\selectfont number of vertices}
\put(-5,20){\rotatebox{90}{ \fontsize{9}{5}\selectfont mean distance}}
\put(75,20){\tiny $\beta\approx 2.35$}
\end{overpic}
\end{center}
\caption{For the graph $G$ in Figure \ref{fig:1} and the rule $r=r_{.91}$, the mean shortest-distance between every pair of vertices in the graphs $G,r(G),...,r^{7}(G)$ is plotted vs. each graph's number of vertices. As this data can be fit quite well by a function of the form $L(j)=\log_{\beta}(c j)$, this suggests that $G$ evolves with the small-world property under the rule $r$.}\label{fig:1a}
\end{figure}

To better establish that specialization under $r=r_p$ leads to graphs (networks) with the real-world properties (i)-(iv) we create a thousand realizations of this process and investigate the statistics of this ensemble of specializations. For each trial we start with a randomly generated undirected, connected graph $G$ with ten vertices and fifteen edges as in Figure \ref{fig:1}, and sequentially specialize the graph using the rule $r_{.91}$. Once the $\ell$th iterate $r^\ell(G)$ has at least a thousand vertices we stop this process.

We then investigate properties (i)-(iii) of the collection of first iterate, the second iterates, and so on, until we reach the collection of twelfth iterates. For property (i), degree distribution, we take the collection of $\ell$th iterates for $\ell=0,1,\dots,12$ and to each of these iterates we again fit a power-law of the form $D(j)=cj^{-\gamma}$ to each graph's degree distribution. The result is a collection of $\gamma$-values from which we find the mean, median, and first and third quartile-values. We plot these four values for each collection of $\ell$th iterates. This is shown in the top-left of Figure \ref{fig:2}. For each power-law fit we compute an $R^2$-value and similarly plot the mean, median, and first and third quartiles for each iterate in the bottom-left of Figure \ref{fig:2}.

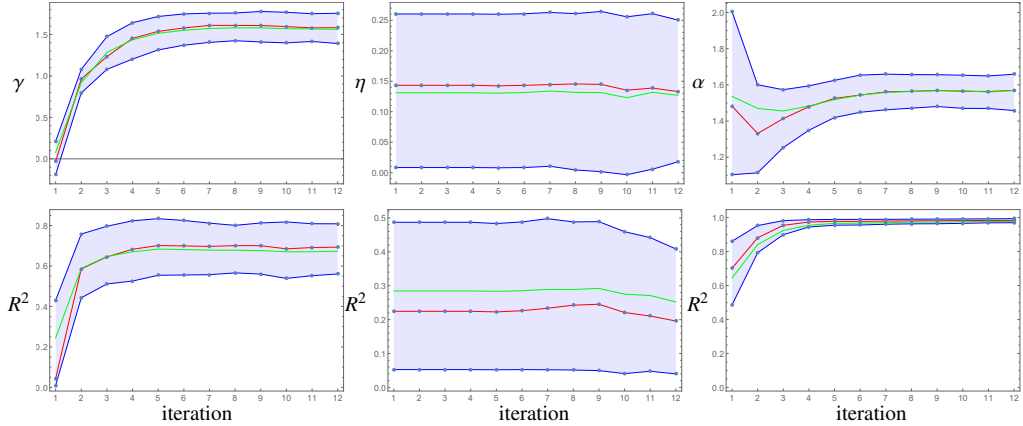
\begin{figure}
\begin{center}
\begin{overpic}[scale=.24]{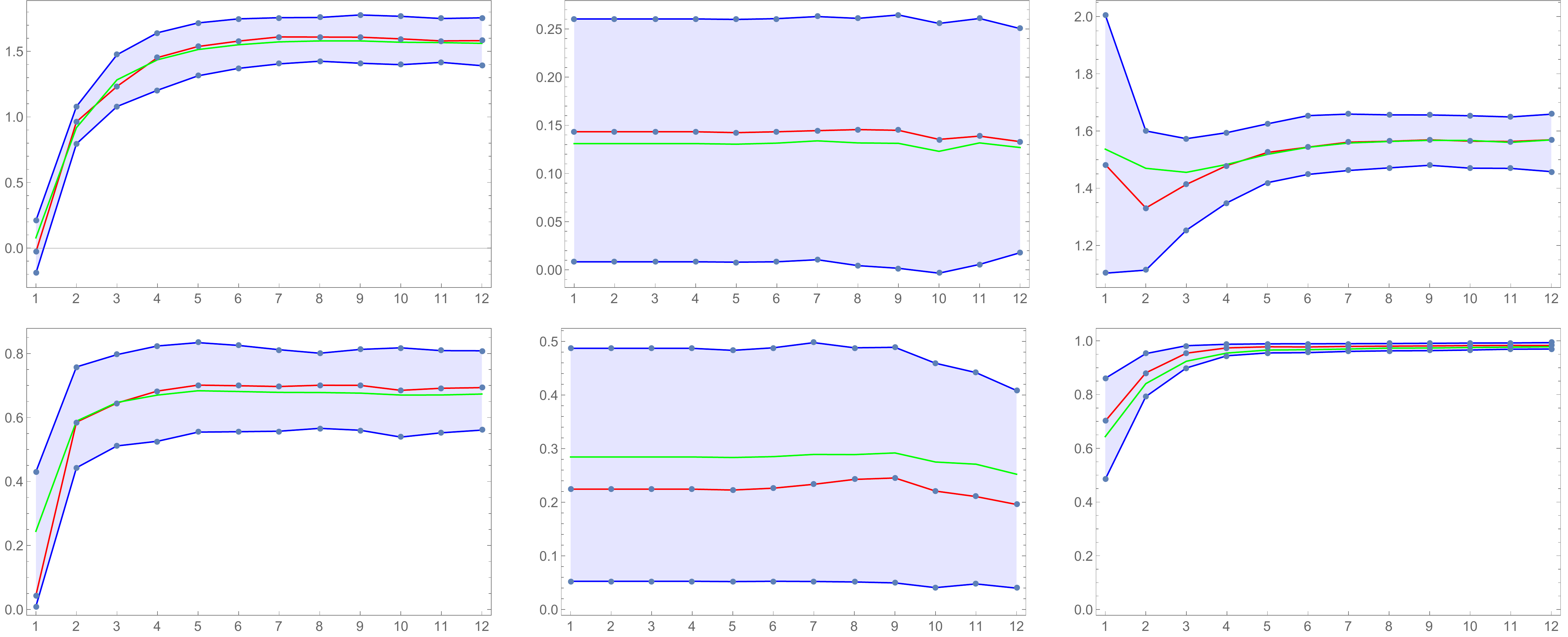}
\put(13,-1.5){\tiny iteration}
\put(47,-1.5){\tiny iteration}
\put(81,-1.5){\tiny iteration}
\put(-2.5,9.5){\tiny $R^2$}
\put(31.5,9.5){\tiny $R^2$}
\put(66,9.5){\tiny $R^2$}
\put(-2,32){\tiny $\gamma$}
\put(32.5,32){\tiny $\eta$}
\put(66.5,32){\tiny $\alpha$}
\end{overpic}
\end{center}
  \caption{Statistical properties are shown for a thousand sequential specializations of a randomly generated graph with 10 vertices and 15 edges under the rule $r=r_{0.91}$. Each graph is sequentially specialized until it has at least one-thousand vertices. The average, mean, and first and third quartiles-values of the quantities $\gamma$, $\eta$, and $\alpha$ representing degree, assortativity, and clustering coefficients, respectively, are shown for each collection of $\ell$th iterates for $\ell=1,\dots, 12$. The mean is shown in green, the median in red, and the first and third quartile-values in blue. The corresponding $R^2$-values are also shown using the same convention.}\label{fig:2}
\end{figure}

It is worth noting that as the graphs grow the average value of $\gamma$ quickly increases to roughly $1.5$ where it stays for the remainder of these iterations. Similarly, the average $R^2$-value rises to roughly $0.7$. Although $R^2$ is not extremely close to $1$, which would indicate a perfect power-law fit, it is nonetheless quite high suggesting that as a graph is repeatedly specialized using the rule $r=r_{.91}$ the result is a graph (network) with a power-law like degree distribution.

Similarly, for property (ii), assortativity, we similarly fit a power-law of the form $A(j)=cj^{-\eta}$ to each graph's mean-neighbor degree distribution. The result is a collection of $\eta$-values from which we again find the mean, median, and first and third quartile-values. We plot these values for each collection of $\ell$th iterates along with the quantities associated with the $R^2$-values. These are shown in the top-middle and bottom-middle of Figure \ref{fig:2}, respectively. For property (iii), the distribution of clustering coefficients, we similarly fit a power-law of the form $C(j)=cj^{-\alpha}$ to each graph's distribution of clustering coefficients and find the average, media, and first and third quartile-values for each collection of iterates. These values along with the associated $R^2$-values are shown in the top-right and bottom-right of Figure \ref{fig:2}, respectively.

These plots indicate that a typical graph develops and/or maintains disassortativity in its mean-neighbor degree distribution as its best-fit line has a nearly constant negative slope under the rule $r$. Additionally, these numerics indicate, especially the $R^2$-values, that a graph quickly develops a clustering coefficient distribution that has a power-law similar what is found in real networks. Hence, in summary, these numerics suggest properties (i)-(iii) of real networks are either developed or maintained by sequentially specializing a graph under the rule $r=r_{.91}$. (Property (iv) is not considered in this case as it is a property of a collection of iterations not a single iterate.)

It is worth noting that the data gathered in Figure \ref{fig:2} is done using the specific value $p=.91$. To further investigate how graphs evolve under the rule $r=r_p$ we also consider a range of $p$-values. In doing so, the idea is to track how properties (i)-(iv) change in response to a change in $p$. For each value of $p$, as before, we start with a random graph $G$ on ten vertices and fifteen edges. We then sequentially specialize using $r=r_p$ until $r^{f}(G)$ has at least one-thousand vertices, so that $r^f(G)$ is the \emph{final} graph in this sequence of specializations. Repeating this process a thousand times we create a thousand ``final" graphs for this particular $p$-value.

Similar to figure \ref{fig:2}, we plot the average, median, and first and third quartile-values of $\gamma=\gamma(p)$, $\eta=\eta(p)$, $\alpha=\alpha(p)$, and $\beta=\beta(p)$ corresponding to properties (i)-(iv), respectively for $p=0.6,0.62,0.64,\dots,0.98$. The quantity $\beta$ is the base of the logarithm $S(j)=\log_{\beta}(c j)$ the \emph{pairwise shortest mean distance plot} of a sequence of graph specializations (cf. Figure \ref{fig:1a}). Hence, we have one-thousand values for each of $\gamma,\eta,\alpha$ and $\beta$ for each value of $p$ we consider.

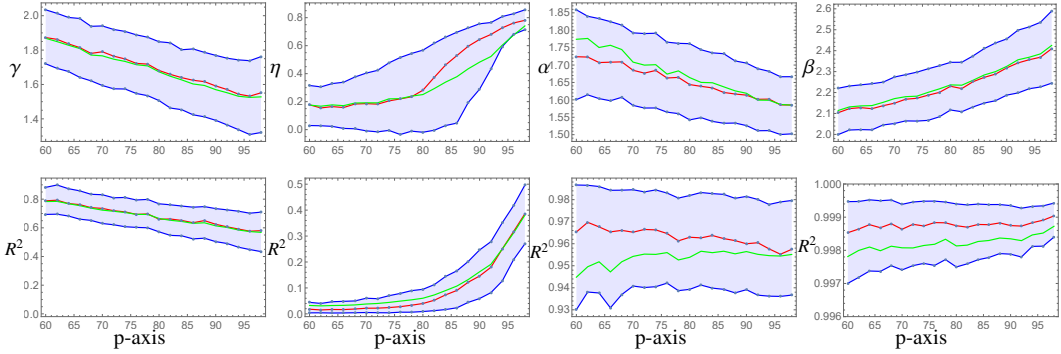
\begin{figure}
\begin{center}
\begin{overpic}[scale=.25]{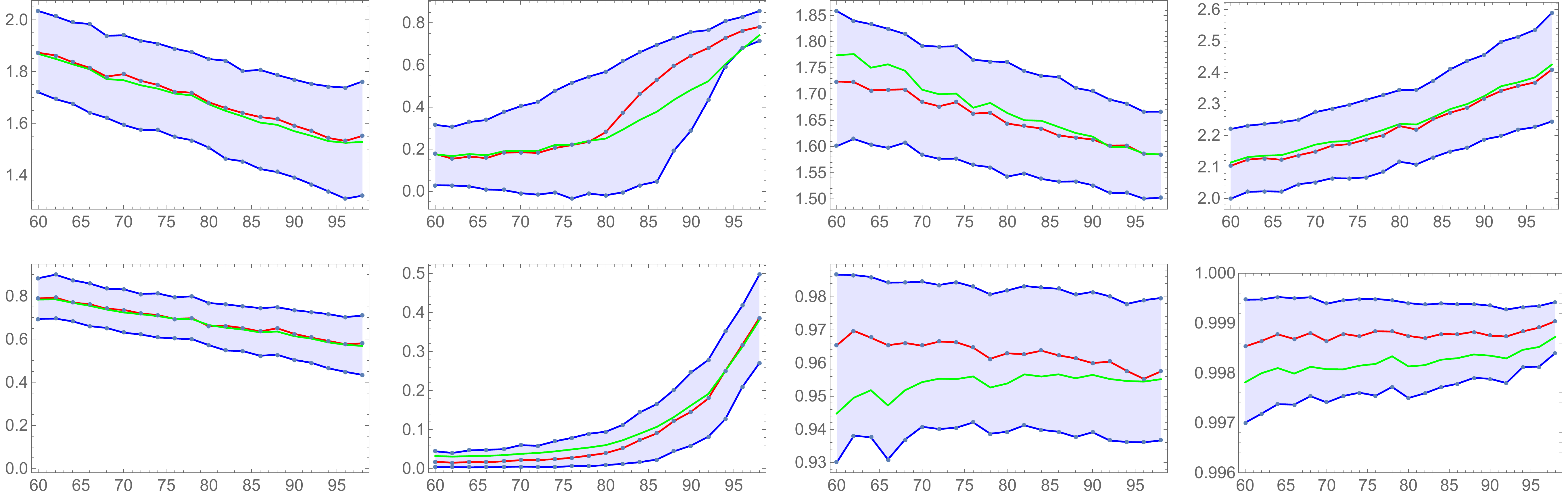}
\put(9,-1){\tiny p-axis}
\put(34,-1){\tiny p-axis}
\put(59,-1){\tiny p-axis}
\put(84,-1){\tiny p-axis}
\put(-1,25){\tiny $\gamma$}
\put(24,25){\tiny $\eta$}
\put(49.7,25){\tiny $\alpha$}
\put(75.2,25){\tiny $\beta$}
\put(-2.2,7.5){\fontsize{7}{5} $R^2$}
\put(23,7.5){\fontsize{7}{5} $R^2$}
\put(48.18,7.2){\fontsize{7}{5} $R^2$}
\put(73.88,7.8){\fontsize{7}{5} $R^2$}
\end{overpic}
\end{center}
  \caption{For the values $p=0.6,0.62,0.64,\dots,0.98$ a thousand graphs are sequentially specialized under the rule $r=r_p$. The
  statistical properties of each collection of specializations are shown above including the average, mean, and first and third quartiles-values of the quantities $\gamma(p)$, $\eta(p)$, $\alpha(p)$, and $\beta(p)$. These represent the degree distribution, assortativity, clustering coefficients, and small-world property, respectively. The mean is shown in green, the median in red, and the first and third quartile-values in blue. The corresponding $R^2$-values are also shown using the same convention.}\label{fig:3}
\end{figure}

The plot of the average, media, and first and third quartile-values for each of $\gamma$, $\eta$, $\alpha$, and $\beta$ are shown in Figure \ref{fig:3}. We also show the corresponding $R^2$-values for each of these approximations. In this figure the quantity $\gamma$ shows a monotonically decreasing trend as does its $R^2$-values. Interestingly, if this trend continues through all values of $p\in(0,1)$ then for $p\in(0,0.45)$ the average value of $\gamma$ in this range will be between 2 and 3, which are the values of $\gamma$ typically observed in many real networks \cite{Newman10}. The reason we do not consider the values $p\in(0,0.6)$ in this simulation is that it is \emph{computationally prohibitive} to specialize graphs over these values. That is, for these $p$-values the graphs grow so quickly in size, the time required to generate a statistically meaningful number of them is quite considerable.

For each of the quantities $\gamma,\eta,\alpha$, and $\beta$ we observe a monotone relationship in $p$, suggesting one can tune the properties of a graph specialized by $r=r_p$ by carefully choosing $p$. Moreover, the $R^2$-values associated with $\gamma,\alpha$, and $\beta$ are quite close to 1 for these $p$-values indicating that properties (i), (iii), and (iv) of graphs generated by $r=r_p$ on average develop these real-world properties as they are specialized under $r=r_p$. Property (ii) does not have a high $R^2$-value but this is not an issue as observations of real networks find only a negative (or positive) trend in a networks mean-neighbor degree distribution. Here we see a negative and therefore disassortative trend for each $p$-value we consider.
\end{example}

\section{Specialization Equivalence}\label{sec4}

As mentioned before Example \ref{ex:rand} there are two types of specialization rule; those that select a unique base vertex subset and those that do not. For instance, the random specialization rule in Example \ref{ex:rand} does not select unique bases. We refer to $\tau$ as a \emph{structural rule} if it does select a unique nonempty subset of vertices from any graph $G$. For instance, $\tau$ could be the rule that selects all vertices with a certain number of neighbors, or eigenvector centrality, etc. (cf. example \ref{ex:evoequ}). An important property of structural rules is that they gives us a way of comparing the topologies of two distinct networks. In particular, any such rule allows us to determine which networks are similar and dissimilar with respect to the rule $\tau$.

To make this precise we say two graphs $G=(V_1,E_1,\omega_1)$ and $H=(V_2,E_2,\omega_2)$ are \emph{isomorphic} if there is a relabeling of the vertices of $V_1$ such that $G=H$ as weighted digraphs. If this is the case, we write $G\simeq H$. The idea is that two graph are similar with respect to a rule $\tau$ if they both evolve to the \emph{same}, i.e. isomorphic graph, under this rule. This allows us to partition all graphs, and therefore networks, into classes of similar graphs with respect to a structural rule $\tau$. This can be stated as the following result.

\begin{theorem}\textbf{(Specialization Equivalence)}\label{thm2}
Suppose $\tau$ is a structural rule. Then $\tau$ induces an equivalence relation $\sim$ on the set of all weighted directed graphs where $G\sim H$ if $\tau(G)\simeq\tau(H)$. If this holds, we call $G$ and $H$ \emph{specialization equivalent} with respect to $\tau$.
\end{theorem}

\begin{proof}
For any $G=(V,E,\omega)$ and structural rule $\tau$ the set $\tau(V)\subseteq V$ is unique implying the graph $\tau(G)=\mathcal{S}_{\tau(V)}(G)$ is unique up to a labeling of its vertices. Clearly, the relation of being specialization equivalent with respect to $\tau$ is reflexive and symmetric. Also, if $\tau(G)\simeq\tau(H)$ and $\tau(H)\simeq\tau(K)$ then there is a relabeling of the vertices of $\tau(G)$ such that $\tau(G)=\tau(H)$ and of $\tau(K)$ such that $\tau(K)=\tau(H)$. Hence, $\tau(G)=\tau(K)$ under some relabeling of the vertices of these graphs implying $\tau(G)\simeq\tau(K)$. This completes the proof.
\end{proof}

Theorem \ref{thm2} states that every structural rule $\tau$ can be used to partition the set of graphs we consider, and by association all networks, into subsets. These subsets, or more formally \emph{equivalence classes}, are those graphs that share a common topology with respect to $\tau$. By \emph{common topology} we mean that graphs in the same class have the same set of component branches and therefore evolve into the same graph under $\tau$.

One reason for studying these equivalence classes is that it may not be obvious, and most often is not, that two different graphs belong to the same class. That is, two graphs may be structurally similar but until both graphs are specialized this similarity may be difficult to see. One of the ideas introduced here is that by choosing an appropriate rule $\tau$ one can discover this similarity as is demonstrated in the following example.

\begin{figure}
\begin{center}
\begin{tabular}{ccc}
    \begin{overpic}[scale=.33,angle=-180.5]{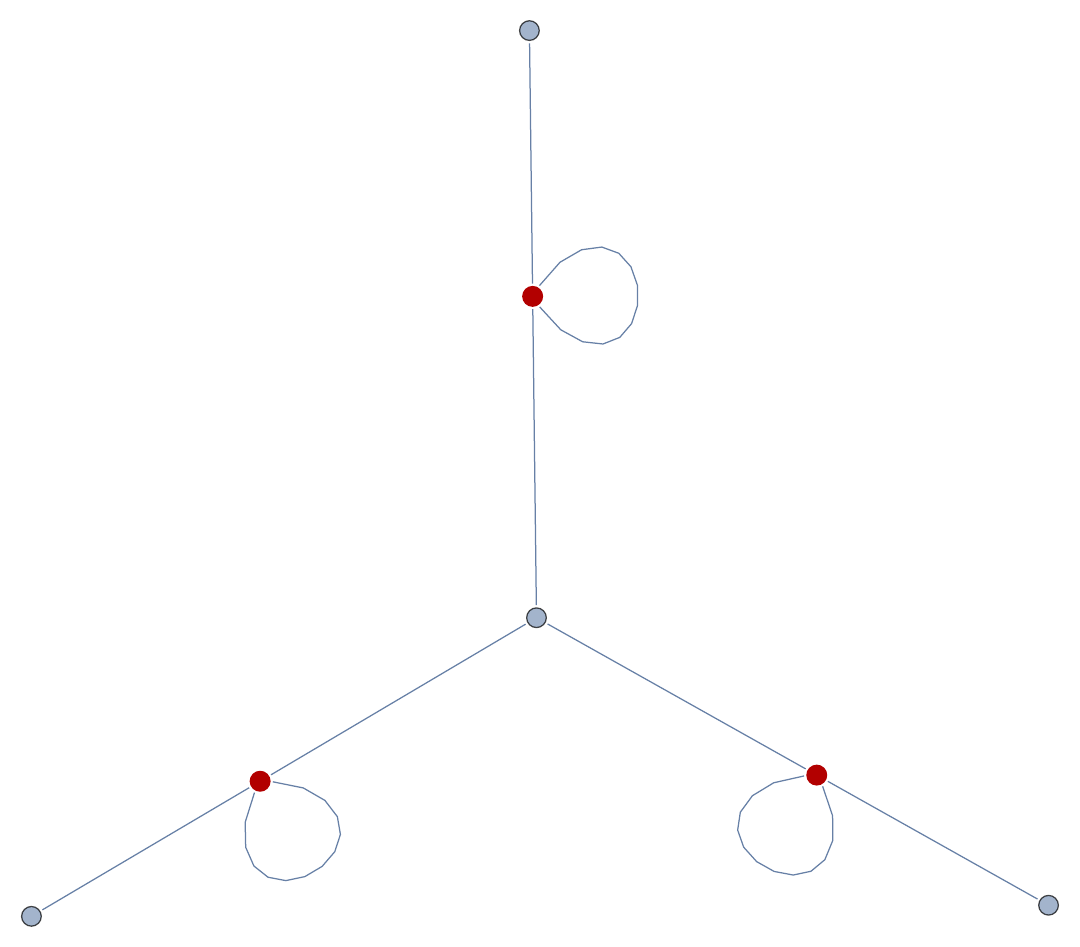}
    \put(38,-5){\put(10,-2){$G$}}
    \end{overpic} &
    \begin{overpic}[scale=.42,angle=181]{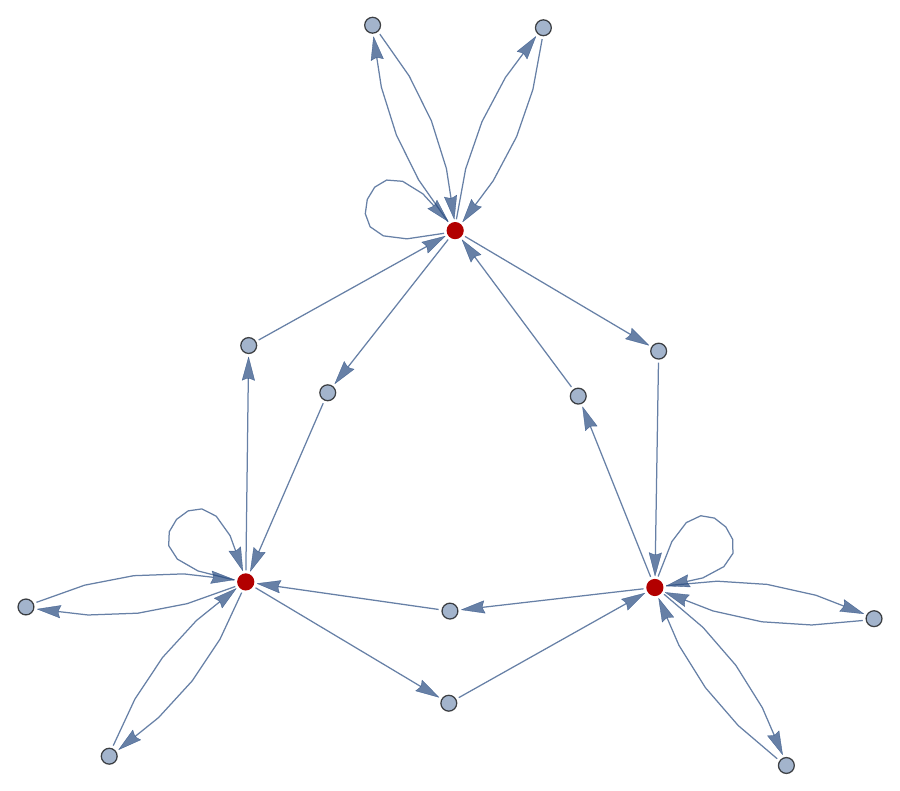}
    \put(25,-7){$\ell(G)\simeq\ell(H)$}
    \end{overpic} &
    \begin{overpic}[scale=.26,,angle=91]{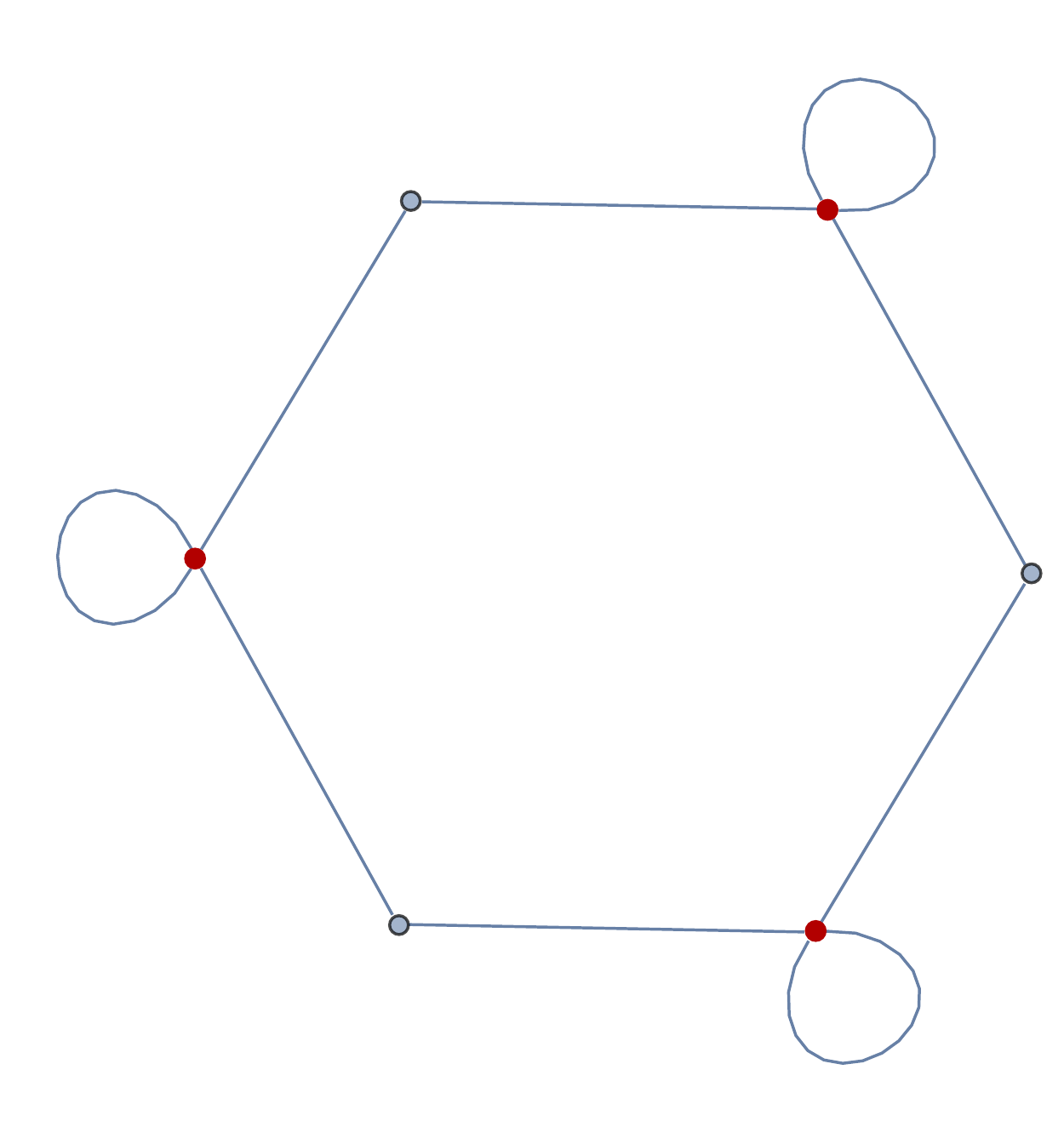}
    \put(46,-7){$H$}
    \end{overpic}
\end{tabular}
\end{center}
  \caption{The graph $G$ and the graph $H$ are specialization equivalent with respect to the rule $\ell$ that selects those vertices of a graph that have loops. That is, the graphs $\ell(G)$ and $\ell(H)$ are isomorphic as is shown.}\label{fig6}
\end{figure}

\begin{example}\label{ex:evoequ} \textbf{(Specialization Equivalent Graphs)}
Consider the unweighted graphs $G=(V_1,E_1)$ and $H=(V_2,E_2)$ shown in Figure \ref{fig6}. Here, we let $\ell$ be the rule that selects all vertices of a graph that have loops, or all vertices if the graph has no loops. The vertices of $G$ and $H$ selected by the rule $\ell$ are the vertices highlighted (red) in Figure \ref{fig6} in $G$ and $H$, respectively. Although $G$ and $H$ appear to be quite different, the graphs $\ell(G)$ and $\ell(H)$ are isomorphic as is shown in Figure \ref{fig6} (center). Hence, $G$ and $H$ belong to the same equivalence class of graphs with respect to the structural rule $\ell$.

It is worth mentioning that two graphs can be equivalent under one rule but not another. For instance, if $w$ is the structural rule that selects vertices \emph{without} loops then $w(G)\not\simeq w(H)$ although $\ell(G)\simeq \ell(H)$.
\end{example}

From a practical point of view, a specialization rule $\tau$ allows those studying a particular class of networks a way of comparing the \emph{specialized topology} of these networks and drawing conclusions about both the specialized and original networks. Of course, the rule $\tau$ should be designed by the particular biologist, chemist, physicist, etc. to have some significance with respect to the networks under consideration.

\section{Concluding Remarks}\label{conc}

In this paper we introduce a class of models of network formation, which we refer to \emph{specialization models} of network growth. These models are based on the observation that most, if not all, real networks specialize the function of their components as they evolve. As a first observation we note that by specializing a network via this model the result is a network whose topology becomes sparser, more modular, and more hierarchical. This is particularly relevant since each of these properties is found throughout networks studied in the biological, technological, and social sciences.

Our method of specialization is highly flexible in that a network can be specialized over any subset of its elements, i.e. any network base. Since so many bases are possible, a natural question is, given a particular real-world network can we find a rule that generates a base that can be used to accurately model this network's growth.

To give evidence that this is possible, we consider the simple rule that randomly selects a certain percentage of a network's elements. We show numerically that this rule evolves the topology of a network in ways that are consistent with the properties widely observed in real networks. This includes (i) having a degree distribution that follows a power-law, i.e. being \emph{scale-free}, (ii) having a \emph{disassortative} property, (iii) having a high \emph{clustering coefficient}, and (iv) having the \emph{small-world property}. So far as the authors know this is the only such model to capture these properties that also creates an increasingly sparse, modular, and hierarchical network topology.

Additionally, we show how certain specialization rules, which we refer to as \emph{structural rules}, can be used to compare the topology of different networks. This notion of similarity, which we refer to as \emph{specialization equivalence}, can be used to partition all networks into those that are similar and dissimilar with respect to a given rule $\tau$. It is worth emphasizing that in practice it is important that this rule be designed by the particular biologist, chemist, physicist, etc. to have some significance with respect to the nature of the networks under consideration.

The notion of specialization and the associated specialization growth model introduced in this paper also lead to a number of open questions, a few of which we mention here. The first, and likely most important, is whether specific specialization rules can be designed to model the growth of specific networks in a way that captures the network's ``finer details," i.e. more than properties (i)-(iv) that are more widely observed in many networks. Such rules will likely be very network dependent and therefore need to be devised and examined again by the particular biologist, chemist, physicist, etc. who has some expertise with the nature of the particular network.

Related to this, properties (i)-(iv) are numerically observed as consequences of specializing using the random rule introduced in example \ref{ex:rand}. It is currently unknown whether these properties can be proven rigorously even for a specific class of initial graphs. More generally, it is unknown for a given rule $\tau$ and graph $G$ what the spectrum of $\tau^k(G)$ is as $k$ increases. Similarly, the topology of the graph $\tau^k(G)$ as $k$ increases is unknown. That is, it is unknown what the graph's \emph{asymptotic spectrum} and \emph{asymptotic topology} are going to be like under $\tau$.

Last, it is worth reiterating that specialization preserves a number of spectral and dynamic properties of the network including the network's eigenvalues, eigenvector centralities, and its dynamic stability under mild conditions. These more technical results are addressed and proven in a following paper.

\section{Acknowledgement} The work of L. A. Bunimovich was partially supported by the NSF grant DMS-1600568. The work of B. Z. Webb is partially supported by was partially supported by the DOD grant HDTRA1-15-0049.


\begin{thebibliography}{9}
\bibitem{KS08} Karlebach, G. $\&$ Shamir, R. Modelling and Analysis of Gene Regulatory Networks. \emph{Nature Reviews Molecular Cell Biology} \textbf{9}, 770-780 (2008).

\bibitem{BO04} Barabasi, A.-L. $\&$ Oltvai, Z. N. Network Biology: Understanding the Cell's Functional Organization. \emph{Nature Reviews Genetics} \textbf{5}, 101-113 (2004).

\bibitem{BS09} Bullmore, E. $\&$ Sporns, O. Complex brain networks: graph theoretical analysis of structural and functional systems. \emph{Nature Reviews Neuroscience} \textbf{10}, 186-198 (2009).

\bibitem{CF14} Clark, R. M. $\&$ Fewell, J. H. Transitioning from unstable to stable colony growth in the desert leafcutter ant Acromyrmex versicolor. \emph{Behavioral Ecology and Sociobiology} \textbf{68} 163-171 (2014).

\bibitem{HEOGF13} Holbrook, C.T., Eriksson, T.H., Overson, R.P., $\&$ J. H. Fewell Colony-size effects on task organization in the harvester ant Pogonomyrmex californicus. \emph{Insectes Sociaux} \textbf{60} 191-201 (2013).

\bibitem{FAIPW12} Fewell, J. H., Armbruster D., Ingraham J., Petersen A., $\&$ Waters, J. S. Basketball Teams as Strategic Networks. \emph{PLOS ONE}, 7(11): e47445 (2012).

\bibitem{GS09} Gross, T. $\&$ Sayama, H. (eds). Adaptive Networks: Theory Models and Applications. Springer (2009)

\bibitem{Bara99} Barabasi, A.-L. $\&$ Albert, R., Emergence of Scaling in Random Networks, Science \textbf{286} 509-512 (1999).

\bibitem{Price76} Price, D.J. de S., A general theory of bibliometric and other cumilitive advantage processes, \emph{J. Amer. Soc. Inform. Sci.} \textbf{27}, 292-306 (1976).

\bibitem{Bara00} Albert, R. $\&$ Barabasi, A.-L., Topology of Evolving Networks: Local events and universiality, \emph{Phys. Rev. Lett.} \textbf{85}, 5234-5237 (2000).

\bibitem{Doro00} Dorogovtsev, S. N., and Mendes, J. F. F., Scaling behaviour of developing and decaying networks, \emph{Europhys Lett.} \textbf{52}, 33-39, (2000).

\bibitem{Krap01} Krapivsky, P. L., Rodgers, G. J., $\&$ Redner, S., Degree distributions of growing networks, \emph{Phys. Rev. Lett.} \textbf{86}, 5401-5404 (2001).

\bibitem{Klein99} Kleinberg, J. M., Kumar, S. R., Raghavan, P., Rajagopalan, and Tomkins, A. The Web as a graph: Measurments, models, and methods, in T. Asano, H. Imai, D. T. Lee, S.-I. Nakano, and T. Tokuyama, eds., \emph{Proceedings of the 5th Annual International Confernce on Combinatorics and Computing}, no. 1627 in Lecture Notes in Computer Science, pp. 1-18, Springer, Berlin (1999).

\bibitem{Sole02} Sole, R. V., Pastor-Satorras, R., Smith, E. $\&$ Kepler, T. B., A model of large scale proteome evolution, Adv. Complex Syst. \textbf{5}, 43-54, (2002).

\bibitem{Vaz03} Vazquez, A., Flammini, A., Martin, A., $\&$ Vespignani, A., Modeling of protein interaction networks, \emph{Complexus} \textbf{1}, 38-44 (2003).

\bibitem{Ferrer03} Ferrer i Cancho, R., $\&$ Sole, R. V., Optimization in complex networks, in R. Pastor-Satorras, J. Rubi, and A. Diaz-Guilera, eds. \emph{Statistical Mechanics of Complex Networks}, no. 625 in Lecture Notes in Physics, pp. 114-125, Springer, Berlin (2003).

\bibitem{Gastner06} Gaster, M. T. $\&$ Newman, M. E. J., Optimal design of spacial distribution networks, Phys. Rev. E 74, 016117 (2006).

\bibitem{Newman10} M. Newman, Networks: An Introduction, Oxford University Press (2010).

\bibitem{Sporns13} O. Sporns, Structure and function of complex brain networks, Dialogues Clin Neurosci. 15(3): 247-262 (2013).

\bibitem{Esp10} C. Espinosa-Soto, A. Wagner Specialization Can Drive the Evolution of Modularity. PLoS Comput Biol 6(3): e1000719. doi:10.1371/journal.pcbi.1000719 (2010).

\bibitem{Milo02} Milo, R., Shen-Orr, S., Itzkovitz, S., Kashtan, N., Chklovskii, D., $\&$ Alon, U. Network Motifs: Simple Building Blocks of Complex Networks, Science \textbf{298}, 824 (2002).

\bibitem{Newman2006} Newman, M. E. J. Modularity and community structure in networks, \emph{Proc Natl Acad Sci USA} 103(23): 8577-8582 (2006).

\bibitem{Clauset08} A. Clauset, C. Moore and M. Newman, Hierarchical structure and the prediction of missing links in networks, Nature \textbf{453}, (2008).

\bibitem{Leskovec2008} J. Leskovec, K. Lang, A. Dasgupta, and M. Mahoney, Statistical Properties of Community Structure in Large Social and Information Networks, Proceedings of the 17th International Conference on World Wide Web, 695--704, ACM, New York, NY, USA (2008).

\bibitem{N03} Newman, M. E. J. The Structure and Function of Complex Networks. \emph{SIAM Review} \textbf{45}, 167-256 (2003).

\bibitem{HG08} Humphries, M. D. $\&$ Gurney, K. Network 'Small-World-Ness': A Quantitative Method for Determining Canonical Network Equivalence. \emph{PLOS ONE}, 3(4): e2051 (2008).

\bibitem{BWBook} Bunimovich, L. A. $\&$ Webb, B. Z. Isospectral Transformations: A New Approach to Analyzing Multidimensional Systems and Networks. Springer Monographs in Mathematics (2014).

\bibitem{BW13} Bunimovich, L. A. $\&$ Webb, B. Z. Restrictions and Stability of Time-Delayed Dynamical Networks. \emph{Nonlinearity} \textbf{26}, 2131-2156 (2013).

\bibitem{Wiki17} https://en.wikipedia.org/wiki/Wikipedia:Disambiguation last accessed Aug. 2017.

\bibitem{Merc17} https://en.wikipedia.org/wiki/Mercury last accessed Aug. 2017.

\bibitem{A07} Alon, U. Network motifs: theory and experimental approaches. \emph{Nature Reviews Genetics} \textbf{8}, 450-461 (2007).

\bibitem{MSA08} MacArthur, B. D., Sanchez-Garcia, R. J., $\&$ Andersonc J. W. Symmetry in Complex Networks. \emph{Discrete Applied Mathematics} \textbf{156}, 3525-3531 (2008).

\bibitem{TSE99} Tononi, G., Sporns, O., $\&$ Edelman G. M. Measures of Degeneracy and Redundancy in Biological Networks, \emph{Proc. Natl. Acad. Sci. USA} \textbf{96} 3257-3262 (1999).

\end{thebibliography}
\end{document}